\pdfoutput=1

\documentclass[sigconf, nonacm]{acmart}

\newcommand\vldbdoi{XX.XX/XXX.XX}

\newcommand\vldbavailabilityurl{URL_TO_YOUR_ARTIFACTS}

\usepackage{algorithmic}
\usepackage{graphicx}
\usepackage{textcomp}
\usepackage{xcolor}
\usepackage{multirow} 
\usepackage{pgfplots}
\usepackage{tikz}
\usepackage{lipsum}
\usepackage{float}
\usepackage{pgfplotstable}
\usepgfplotslibrary{groupplots}
\usepackage{multirow}
\usepackage{tabularx}
\usepackage{subfigure}
\usepackage{graphicx}
\usepackage{textcomp}
\usepackage{color}
\usepackage{listings}
\usepackage{booktabs}
\usepackage{makecell}
\usepackage{balance}
\usetikzlibrary{patterns}
\usepackage{tikz}
\usepackage{pgfplots}
\usetikzlibrary{patterns}
\usepackage{multirow}
\usepackage{graphicx} 
\usepackage{svg} 
\usepackage{multirow}
\usepackage[normalem]{ulem}
\useunder{\uline}{\ul}{}

\usepackage{array}
\usepackage{pgfplots}
\usepackage{subcaption}
\pgfplotsset{compat=1.17} 
\usepackage{caption}
\usepackage{adjustbox}
\usepackage{calc}
\usepackage{multirow}
\usepackage{algorithm}

\usepackage[font=small,skip=0pt]{caption}
\usepackage{enumitem}
\setlist[itemize]{noitemsep, topsep=2pt}

\usepackage{amsmath} 
\DeclareMathOperator*{\argmin}{arg\,min} 

\pgfplotsset{compat=1.16}  
\usetikzlibrary{patterns} 
\definecolor{c1}{RGB}{42,99,172} %
\definecolor{c2}{RGB}{255,88,93}
\definecolor{c3}{RGB}{255,181,73}
\definecolor{c4}{RGB}{119,71,64} %
\definecolor{c5}{RGB}{228,123,121} %
\definecolor{c6}{RGB}{208,167,39} %
\definecolor{c7}{RGB}{0,51,153}
\definecolor{c8}{RGB}{56,140,139} 
\definecolor{c9}{RGB}{0,0,0} 

\definecolor{color1}{HTML}{FFF4D8}
\definecolor{color2}{HTML}{FF7163}
\definecolor{Blood}{HTML}{860309}
\definecolor{Olive}{HTML}{807805}
\definecolor{Topaz}{HTML}{F5C678}
\definecolor{American Yellow}{HTML}{F3AA07}
\definecolor{Giants Orange}{HTML}{FF5C1E}
\definecolor{Persian Plum}{HTML}{6C1D2A}
\definecolor{Pearl Aqua}{HTML}{85D6B2}
\definecolor{Moonstone}{HTML}{3F9EBD}
\definecolor{Blue Jeans}{HTML}{5BBCF0}
\definecolor{St. Patrick's Blue}{HTML}{2B2D7C}
\definecolor{Green}{HTML}{159879}
\definecolor{Green2}{HTML}{31D683}
\definecolor{Purple}{HTML}{C1689C}
\pgfdeclarepatternformonly{myCrossHatch}{\pgfqpoint{-1pt}{-1pt}}{\pgfqpoint{3pt}{3pt}}{\pgfqpoint{4pt}{4pt}}%
{
  \pgfsetlinewidth{1pt} 
  \pgfpathmoveto{\pgfqpoint{0pt}{0pt}}
  \pgfpathlineto{\pgfqpoint{4pt}{4pt}}
  \pgfpathmoveto{\pgfqpoint{4pt}{0pt}}
  \pgfpathlineto{\pgfqpoint{0pt}{4pt}}
  \pgfusepath{stroke}
}

\pgfdeclarepatternformonly{myDots}{\pgfqpoint{-1pt}{-1pt}}{\pgfqpoint{2.5pt}{2.5pt}}{\pgfqpoint{2.5pt}{2.5pt}}%
{
  \pgfpathcircle{\pgfqpoint{1pt}{0.8pt}}{0.8pt} 
  \pgfusepath{fill}
}

\pgfdeclarepatternformonly{myDenseLines}{\pgfqpoint{-1pt}{-1pt}}{\pgfqpoint{2pt}{2pt}}{\pgfqpoint{1.8pt}{1.8pt}}%
{
  \pgfsetlinewidth{1pt} 
  \pgfpathmoveto{\pgfqpoint{0pt}{0pt}}
  \pgfpathlineto{\pgfqpoint{4pt}{4pt}}
  \pgfusepath{stroke}
}

\pgfdeclarepatternformonly{myDenseLines2}{\pgfqpoint{0pt}{0pt}}{\pgfqpoint{10pt}{10pt}}{\pgfqpoint{10pt}{10pt}}%
{
  \pgfsetlinewidth{0.5pt} 
  \pgfpathmoveto{\pgfqpoint{0pt}{2pt}}
  \pgfpathlineto{\pgfqpoint{10pt}{10pt}}
  \pgfusepath{stroke}
}

\pgfdeclarepatternformonly{myContinuousLine}{\pgfqpoint{-2pt}{-2pt}}{\pgfqpoint{4pt}{4pt}}{\pgfqpoint{2pt}{2pt}}%
{
  \pgfsetlinewidth{1pt} 
  \pgfpathmoveto{\pgfqpoint{-2pt}{-2pt}}
  \pgfpathlineto{\pgfqpoint{6pt}{6pt}}
  \pgfusepath{stroke}
}

\newcommand{\tbase}{TierBase }

\newcommand{\sstitle}[1]{\vspace{0ex} \noindent\underline{{\it #1}}}

\theoremstyle{definition}
\newtheorem{definition}{Definition}
\def\BibTeX{{\rm B\kern-.05em{\sc i\kern-.025em b}\kern-.08em
    T\kern-.1667em\lower.7ex\hbox{E}\kern-.125emX}}

\begin{document}
\title{ TierBase: A Workload-Driven Cost-Optimized Key-Value Store [Extended Version] }


\author{
Zhitao Shen$^{\dagger}$, Shiyu Yang$^{\S}$, Weibo Chen$^{\S \dagger }$,  Kunming Wang$^{\S \dagger}$, Yue Li$^{\dagger}$, Jiabao Jin$^{\dagger}$, Wei Jia$^{\dagger}$, Junwei Chen$^{\dagger}$,  Yuan Su$^{\dagger}$, Xiaoxia Duan$^{\dagger}$, Wei Chen$^{\dagger}$, Lei Wang$^{\dagger}$, Jie Song$^{\dagger}$, Ruoyi Ruan$^{\dagger}$, Xuemin Lin$^{\ddagger}$ \\
}

\renewcommand{\shortauthors}{Shen et al.}
\affiliation{$^\dagger$Ant Group ; ~ $^\S$Guangzhou University ; ~ $^{\ddagger}$ Shanghai Jiao Tong University}
\affiliation{ \{zhitao.szt,~ chenweibo.cwb,~ wangkunming.wkm,~ ly321766,~ jinjiabao.jjb, ~ jw94525,~ weisong.cjw,~ stevensu.sy,~ \\ 
xiaoxia.dxx,~ cw281808,~ wl177541,~ peter.sj,~ ruoyi.ruanry\}@antgroup.com; syyang@gzhu.edu.cn; xuemin.lin@sjtu.edu.cn}


\begin{abstract}

In the current era of data-intensive applications, the demand for high-performance, cost-effective storage solutions is paramount. This paper introduces a Space-Performance Cost Model for key-value store, designed to guide cost-effective storage configuration decisions. The model quantifies the trade-offs between performance and storage costs, providing a framework for optimizing resource allocation in large-scale data serving environments. Guided by this cost model, we present TierBase, a distributed key-value store developed by Ant Group that optimizes total cost by strategically synchronizing data between cache and storage tiers, maximizing resource utilization and effectively handling skewed workloads. To enhance cost-efficiency, TierBase incorporates several optimization techniques, including pre-trained data compression, elastic threading mechanisms, and the utilization of persistent memory.  We detail TierBase's architecture, key components, and the implementation of cost optimization strategies. Extensive evaluations using both synthetic benchmarks and real-world workloads demonstrate TierBase's superior cost-effectiveness compared to existing solutions. Furthermore, case studies from Ant Group's production environments showcase TierBase's ability to achieve up to 62\% cost reduction in primary scenarios, highlighting its practical impact in large-scale online data serving.
\end{abstract}

\maketitle

\begingroup
\renewcommand\thefootnote{}\footnote{\noindent
}\addtocounter{footnote}{-1}\endgroup


\section{Introduction}

In the current era of data-intensive applications, the demand for high-performance, cost-effective storage solutions is paramount. Key-value store, with their promise of scalability, flexibility, and rapid data access, have emerged as a pivotal component in this landscape. Renowned key-value stores such as Redis \cite{redis}, Memcached \cite{memcached}, Cassandra \cite{Cassandra}, and HBase \cite{hadoop} have set industry benchmarks, offering solutions designed for diverse use cases ranging from in-memory caching to persistent storage.

At Ant Group, we face numerous challenges in managing our online data serving systems. These challenges stem from the intrinsic nature of our services and the evolving demands of our vast user base. We handle \textbf{\textit{the immense volume of data}} generated by billions of users, necessitating robust, low-latency, and cost-effective storage solutions.  The diverse scenarios and applications lead to \textbf{\textit{a wide spectrum of workloads}} with varying needs for reliability, durability, and latency.  Additionally, \textbf{\textit{significant skewness in data access patterns}}, with dynamically changing hot spots, complicates data access and caching strategies. 
Furthermore, in order to handle these workloads, we have to maintain an extremely large number of machines, which leads to a \textbf{\textit{non-trivial configuration decision}} challenge. This challenge demands a cost model to accurately quantify the cost-performance trade-offs.

To address these challenges, we introduce TierBase, a distributed key-value store developed by Ant Group since 2017. Initially Redis-compatible, TierBase has evolved to support advanced functions such as CAS operations, wide-columns, and vector searching. Extensively utilized within Ant Group, TierBase maintains sub-millisecond access latency even under peak loads of hundreds of millions of queries per second (QPS), crucial for delivering seamless user experiences during events like the Double 11 shopping festival.

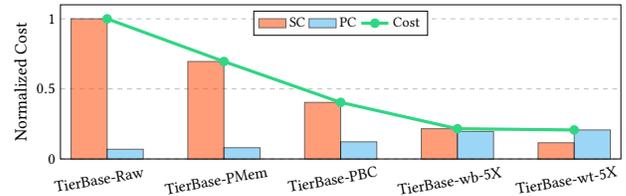
\begin{figure}[t]
    \centering
    \begin{tikzpicture}[scale=0.6]
        \begin{axis}[
            width=14cm, 
            height=5cm, 
            bar width=5cm, 
            ymin=0,
            ymax=1.1, 
            bar width=0.8cm, 
            ylabel={Normalized Cost}, 
            symbolic x coords={TierBase-Raw, TierBase-PMem, TierBase-PBC, TierBase-wb-5X, TierBase-wt-5X}, 
            xtick=data, 
            legend style={at={(0.5,0.8)}, anchor=south,legend columns=-1}, 
            ymajorgrids=true, 
            ylabel style={font=\Large}, 
            xlabel style={font=\huge},
            x tick label style={font=\large, rotate=10,  xshift=-5pt}, %
            y tick label style={font=\large}, 
            grid style=dashed,
            tick style={draw=none} 
        ]
            \addplot[ybar, fill=Giants Orange, bar shift=-0.4cm, area legend, opacity=0.6] coordinates {
                (TierBase-Raw, 1.0) (TierBase-PMem, 0.695421468) (TierBase-PBC, 0.403577926) (TierBase-wb-5X, 0.21573681) (TierBase-wt-5X, 0.11552456)};
            \addlegendentry{SC}
            
            \addplot[ybar, fill=Blue Jeans, bar shift=0.4cm, area legend, opacity=0.6] coordinates {
                (TierBase-Raw, 0.0693208) (TierBase-PMem, 0.080181928) (TierBase-PBC, 0.122255913) (TierBase-wb-5X, 0.196112796) (TierBase-wt-5X, 0.206943602)};
            \addlegendentry{PC}
            
            \addplot[Green2, line width=2pt, mark=*] coordinates {
                (TierBase-Raw, 1.0) (TierBase-PMem, 0.695421468) (TierBase-PBC, 0.403577926) (TierBase-wb-5X, 0.21573681) (TierBase-wt-5X, 0.206943602)};
            \addlegendentry{Cost}
        \end{axis}
    \end{tikzpicture}
    \caption{Cost comparison in TierBase}
    \label{totalcost}
\end{figure}

To optimize costs further, we developed several cost-saving strategies within TierBase. Figure \ref{totalcost} illustrates the cost-performance trade-offs and the impact of these optimization techniques in a real-world use case. The figure shows how enabling our cost-saving strategies leads to reductions in both space cost ($SC$) and performance cost ($PC$), which represent storage and query processing expenses, respectively. In our cost model, the overall cost is defined as the maximum of $SC$ and $PC$, represented by the green line in the figure. Despite the increase in $PC$, the significant reduction in $SC$ results in an overall cost decrease. Our pre-trained compression technique (TierBase-PBC) achieves up to a 62\% cost reduction over the baseline (TierBase-Raw), demonstrating substantial savings in one of our primary online serving scenarios.

In developing TierBase, we focus on two key questions:

\noindent \textbf{Q1: How can we develop a comprehensive cost model for large-scale online data serving systems that adapts to varying workloads?}
Developing a comprehensive cost model for complex online data serving systems requires identifying key metrics that accurately capture real-world costs across various workloads. Traditional models \cite{tco,10.5555/1287369.1287387} typically focus on overall system costs without considering specific workload characteristics. The challenge lies in developing a quantitative framework that accurately models cost-performance trade-offs between different storage configurations while incorporating workload-specific characteristics and capturing the non-linear relationship between system configuration and cost.

\noindent \textbf{Q2: How can we effectively evaluate and apply optimization techniques in key-value stores to balance performance and cost for specific workloads?}
The research and industry communities have developed a number of innovative techniques \cite{SlimCache,riak} to optimize performance and storage efficiency for key-value stores. However, they are often workload-specific, and there is no universal solution. The challenge lies in developing a unified framework to evaluate and apply these diverse optimization techniques across different workloads considering both performance metrics and overall system cost. This cost model aims to enable informed decisions about which techniques to apply in different scenarios, moving beyond one-size-fits-all solutions towards a workload-aware optimization strategy for key-value stores.

To answer these questions, in this paper, we introduce: 
\noindent\textbf{Space-Performance Cost Model}: A novel approach balances performance and space costs, proposing that optimal cost is achieved when these factors are equal. This model extends to tiered storage systems, incorporating cache ratios and miss ratios to determine cost-effectiveness. 

\noindent\textbf{TierBase}: Guided by this cost model, TierBase employs a tiered storage architecture that balances performance and storage space for cost-efficiency. It incorporates innovative features such as a flexible design for handling diverse workloads, pre-trained data compression techniques, elastic threading mechanisms, and utilization of persistent memory.

\noindent\textbf{Cost Optimization Framework}: We introduce a framework for evaluating and optimizing costs for key-value stores. This framework includes strategies for adapting to diverse workloads and provides guidelines for making cost-effective decisions in system configuration and resource allocation.

To sum up, the main contributions are as follows:
\begin{itemize}
\item We propose a comprehensive Space-Performance Cost Model for key-value store that aligns with and extends the classic Five-Minute Rule, guiding decisions on optimal storage configuration selection.
\item We present the architecture of TierBase, a distributed key-value store that leverages a tiered storage design to effectively balance performance and cost-efficiency across diverse workloads.
\item We introduce and evaluate several cost-optimizing techniques implemented in TierBase, including pre-trained data compression, elastic threading mechanisms, and the utilization of persistent memory.
\item We demonstrate the effectiveness of our cost optimization framework through extensive evaluations using synthetic benchmarks, real-world workloads, and case studies from Ant Group's production environments. These evaluations showcase TierBase's superior cost-effectiveness compared to existing solutions and its ability to achieve significant cost reductions in large-scale, data-intensive applications.
\end{itemize}

\section{Space-performance Cost Modeling}\label{sec:cost_model}

In this section, we introduce a novel cost model for key-value storage systems, based on real-world implementations at Ant Group. Our model unifies performance and storage costs, providing a comprehensive framework for system optimization in online data serving with real-time latency requirements.

\begin{figure}[t]
 
  \centering
  \subfigure[Single-Tier Storage]{
    \centering
    \includegraphics[scale=0.29]{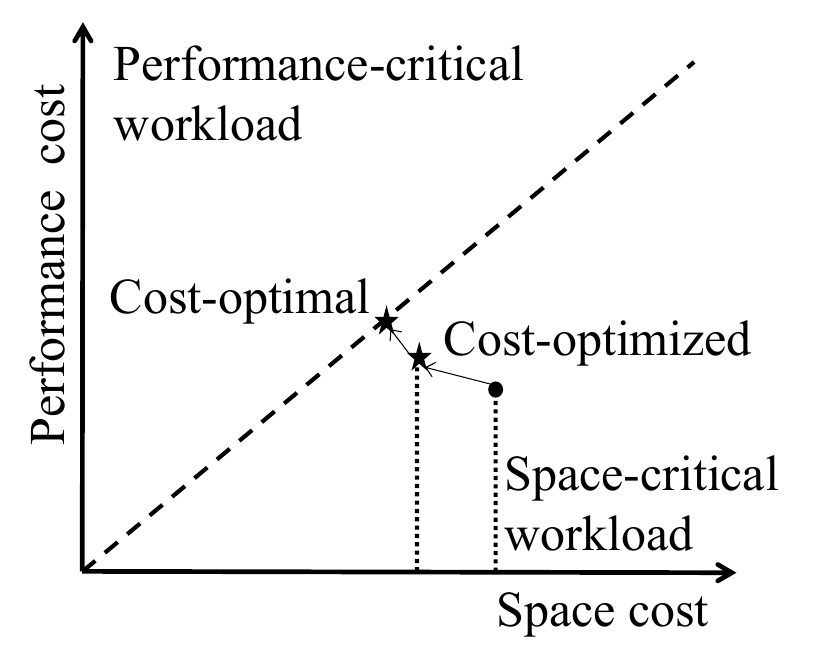}
    \label{fig:cost-model1}
  }
  \hspace{-0.2cm}
  \subfigure[Tiered Storage]{
    \centering
    \includegraphics[scale=0.29]{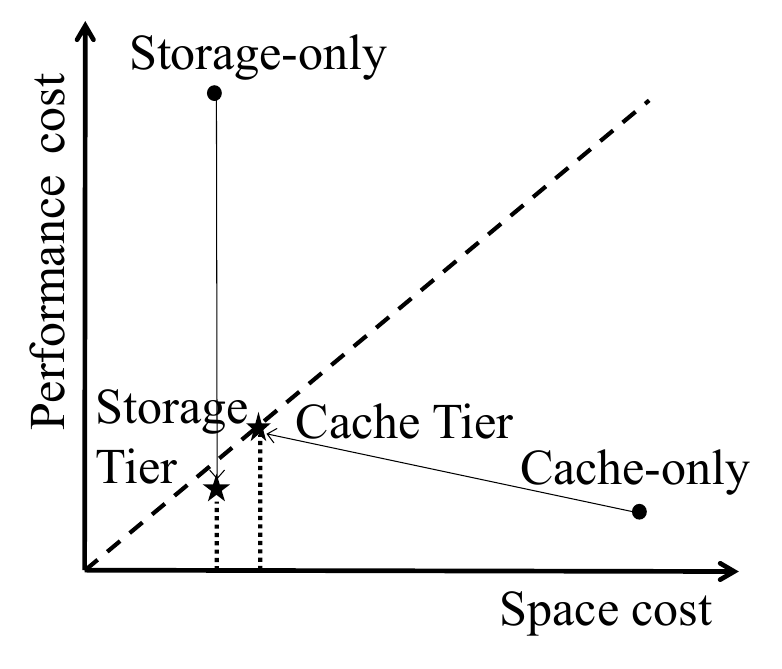}
    \label{fig:cost-model2}
  }
  
  \caption{Space-Performance Cost Model}
  \label{fig:cost-model}
  
\end{figure}

\subsection{Cost Analysis Framework}

Our framework is built upon two primary components:
\noindent\textbf{Performance Cost ($PC$)} in key-value storage systems reflects the expenses associated with data transfer from storage media to end-users, encompassing resource utilization for both read and write operations. This includes CPU overhead, network I/O, disk IOPS consumption, and memory bandwidth usage.

\noindent\textbf{Space Cost ($SC$)} depends on the resources and expenses for data storage, varying by storage medium and space used. In caching and in-memory storage, it's linked to the data volume in RAM, which is fast but expensive. Disk-based storage costs are typically cheaper for large data volumes. These costs also account for storage overhead related to data structures and replicas necessary to ensure reliability and availability.

\noindent\textbf{Space-Performance Cost Model.}
Our Space-Performance Cost Model is based on the observation that in enterprise data centers and cloud environments, resource instances (virtual machines or containers with compute and storage resources) are typically provided with pre-defined allocations. These allocations are usually evenly divided to maximize utilization, precluding arbitrary resource allocation.

Given this constraint, for a given workload $w$, which is a stream of read and write operations on a given resource instance $i$ and storage system with configuration $s$, both the maximum performance ($MaxPerf(w,i,s)$) and the maximum storable data amount ($MaxSpace(w,i,s)$) are deterministic and quantifiable. These metrics are measured in queries per second (QPS) and gigabytes (GB), respectively. 

In a distributed, shared-nothing architecture, we define the monetary cost $C$ as the maximum of the performance cost ($PC$) and the space cost ($SC$) for a given workload on a set of resource instances $i$ with the same configuration:

\begin{definition}[Cost of workload $w$]
The cost of workload $w$ is defined as the maximum of the performance cost and the space cost for a given workload on a set of resource instances $i$ with a specific storage configuration $s$:

\begin{equation}
C(w, i, s) = \max(PC(w, i, s), SC(w, i, s))
\end{equation}

Where:
\begin{align*}
PC(w, i, s) &= Cost(i) \times \left\lceil \frac{QPS(w)}{MaxPerf(w, i, s)} \right\rceil \\
SC(w, i, s) &= Cost(i) \times \left\lceil \frac{DataSize(w)}{MaxSpace(w, i, s)} \right\rceil
\end{align*}

Here, $Cost(i)$ is the monetary cost of a single resource instance $i$, $QPS(w)$ is the total queries per second for the workload $w$, $DataSize(w)$ is the total amount of data to be stored for workload $w$, and $MaxPerf(w, i, s)$ and $MaxSpace(w, i, s)$ are the maximum performance and space capacity for the given resource instance $i$ and storage configuration $s$, respectively.
\end{definition}

In a distributed, shared-nothing architecture, the maximum of the performance cost and space cost is used because the system must be provisioned to meet the greater of the two demands: query processing or data storage. This ensures that sufficient resources are allocated to handle the workload's heaviest requirement, whether it's the query throughput or the data volume. 

For real-world deployments, we incorporate tolerance ratios for both $MaxPerf$ and $MaxSpace$, ensuring system redundancy and reliability. These ratios accommodate variations in workload distribution and access patterns, enabling adaptation to scenarios that deviate from the cost model's assumption of even data sharding.

In Figure \ref{fig:cost-model1}, we illustrate both the performance and space costs. For workloads where the performance cost exceeds the space cost, we categorize them as \textit{performance-critical workloads}. Conversely, when the space cost is dominant, we refer to these as \textit{space-critical workloads}.

\subsection{Cost Efficiency Metrics}\label{sec:cost_mertrics}

Our cost model provides a framework for optimizing resource allocation and system configuration in distributed key-value storage systems. By leveraging the insights from this model, we can make informed decisions to minimize overall costs while meeting both performance and space requirements.

Consider a scenario where a workload typically requires multiple resource instances. We can simplify our model by removing the ceiling function and define the cost metrics as follows:

\begin{definition}[Cost Metrics]
We define two key cost metrics:
\begin{align*}
CPQPS &= Cost(i) / MaxPerf(w,i,s) \\
CPGB &= Cost(i) / MaxSpace(w,i,s)
\end{align*}
where CPQPS is the Cost per Query per Second, representing the performance cost incurred for processing each query per second, and CPGB is the Cost per GB, representing the space cost for storing each gigabyte of data.
\end{definition}

For ease of presentation in subsequent discussions, we will consider the same $i$ and $w$, allowing us to represent different configurations as subscripts in later formulae.
Using these cost metrics, we can express the total cost of the storage system as:

\begin{equation}
C = \max(CPQPS \times QPS, CPGB \times DataSize)
\end{equation}

This formulation allows us to evaluate and optimize the overall system cost by considering both performance and space requirements.

\subsection{Space-Performance Trade-off and Optimal Cost Theorem}

Our cost model reveals an inherent trade-off between performance and space costs in key-value storage systems. This trade-off forms the basis for our Optimal Cost Theorem.

\begin{definition}[Space-Performance Trade-off of Storage Configurations]
Given a set of storage configurations ${S}$, the space-performance trade-off describes the relationship between the Cost per Query per Second ($CPQPS$) and the Cost per Gigabyte ($CPGB$), expressed as: $CPQPS_s = f(CPGB_s), \quad s \in S$
where $f$ is a non-increasing function. As $CPGB_s$ decreases for a given configuration $s$, $CPQPS_s$ tends to increase, and vice versa.
\end{definition}

A typical example of this trade-off is data compression: for a fixed compression algorithm, setting higher compression levels reduces space cost ($CPGB$) but increases performance cost ($CPQPS$) due to added computational overhead, allowing a trade-off between storage efficiency and query performance.

Based on the trade-off, we establish the Optimal Cost Theorem:

\begin{theorem}[Optimal Cost $C^*$]
For a given workload $w$ with requirements $QPS$ and $DataSize$, and a set of storage configurations $S$, the optimal cost $C^*$ is achieved by selecting the configuration $s^* \in S$ that minimizes the overall cost while balancing performance and space costs:
$C^* = \min_{s \in S} \max(PC_{s}, SC_{s})$

The optimal configuration $s^*$ is one that minimizes the absolute difference between performance and space costs:
$s^* = \argmin_{s \in S} |PC_{s}  - SC_{s}|$
\end{theorem}

\begin{proof}
Let $s^*$ be the optimal configuration that minimizes the overall cost:

$C^* = \min_{s \in S} \max(PC_s, SC_s)$

Assume, for contradiction, that $PC_{s^*} \neq SC_{s^*}$. Without loss of generality, let $PC_{s^*} > SC_{s^*}$. Then:

$C^* = PC_{s^*}$

Now, consider a configuration $s'$ that slightly reduces $PC_{s^*}$ at the expense of increasing $SC_{s^*}$, such that:

$PC_{s'} = PC_{s^*} - \epsilon$
$SC_{s'} = SC_{s^*} + \delta$

Where $\epsilon > 0$ and $\delta > 0$ are small positive values.

If we choose $\epsilon$ and $\delta$ such that $SC_{s^*} + \delta < PC_{s^*} - \epsilon$, then:

$\max(PC_{s'}, SC_{s'}) = PC_{s'} = PC_{s^*} - \epsilon < PC_{s^*} = C^*$

This contradicts the assumption that $s^*$ is the optimal configuration. Therefore, our initial assumption must be false, and we must have $PC_{s^*} = SC_{s^*}$ for the optimal configuration.

Hence, the optimal configuration $s^*$ is one that minimizes the absolute difference between performance and space costs:

$s^* = \argmin_{s \in S} |PC_s - SC_s|$

This proves both parts of the theorem.
\end{proof}

The proof demonstrates that any imbalance between performance and space costs can be optimized to yield a lower total cost, thus establishing the optimal point at their equality.

This theorem provides a guiding principle for optimizing key-value storage systems. It suggests that the most cost-effective configuration is one where the system's resources are balanced such that neither space nor performance costs dominate. This balance point represents the optimal trade-off between space and performance for the given workload.

\subsection{Cost Model with Tiered Storage}\label{cost_tired}

Tiered storage systems typically comprise two layers: a cache tier and a storage tier. The cache tier, often utilizing memory-based technologies, prioritizes performance, while the storage tier focuses on capacity and cost-effectiveness. This structure allows systems to balance performance and capacity requirements more efficiently than single-tier alternatives.

In a tiered storage system, the total cost is a function of both the cache tier and the storage tier. We propose a comprehensive cost model for tiered storage systems that accounts for both performance and capacity costs across tiers:

\begin{align}
C_{tiered} = \max(&PC_{cache} 
+ PC_{miss}\times MR, \nonumber \\
&SC_{cache} \times CR) \nonumber \\
+ \max(&PC_{storage} \times MR, SC_{storage})
\label{eq:tiered_cost}
\end{align}

where:
\begin{itemize}
\item $CR$ is the cache ratio (cache capacity / total capacity)
\item $MR$ is the cache miss ratio (proportion of requests served by the storage tier)
\item $PC_{miss}$ is the additional performance cost incurred on a cache miss
\item $PC_{cache/storage}$ and $SC_{cache/storage}$ represent performance and space costs for each tier
\end{itemize}

This model enables a comparison between tiered storage and single-tier alternatives. Tiered storage becomes cost-effective when its total cost is lower than both a pure cache solution and a pure storage solution, expressed as: $C_{tiered} < \min(C_{cache}, C_{storage})$. This approach can be particularly effective for workloads with skewed data access patterns.

The model provides insights into optimizing cache ratio ($CR$) and managing miss ratio ($MR$) to minimize overall system cost. Further detailed cost analysis on tiered storage is presented in Section \ref{cost_analysis_tired_storage}.

\subsection{Cost Optimization Strategies}

Our cost model serves as a valuable guide for optimization efforts, assisting system designers and administrators in their decision-making process. We discuss optimization strategies for both single-tier and tiered storage systems, using the space-performance cost model as our framework.

\subsubsection{Single-Tier Storage Optimization}
In single-tier storage systems, the primary goal is to balance performance cost ($CPQPS$) and space cost ($CPGB$) to minimize overall cost. The optimization strategies depend on the workload characteristics:

\sstitle{Space-Critical Workloads}: 
When the workload is space-critical. Here, reducing $CPGB$ becomes the primary goal. 
One approach is to enable data compression, which reduces $CPGB$ but may increase $CPQPS$ due to the additional computational overhead. The overall cost can still be optimized through such trade-offs according to the Optimal Cost Theorem. Other potential approaches include using instance with larger storage and implementing tiered storage solutions.

\sstitle{Performance-Critical Workloads}: In this case, optimization efforts should focus on reducing $CPQPS$. Strategies may include:
Optimizing query execution, tiered caching mechanisms, and utilizing faster storage media for frequently accessed data.

\sstitle{Resource Instance Selection}: 
Choosing the right type of resource instance can significantly impact both $CPQPS$ and $CPGB$. This involves analyzing different instance types to find the optimal balance between performance capabilities and storage space for the specific workload.

\subsubsection{Tiered Storage Optimization}

Tiered storage systems offer additional opportunities for cost optimization by leveraging the strengths of different storage tiers. The effectiveness of tiered storage depends on three key factors:

\sstitle{Skewed Data Access Pattern:}
Optimal tiered storage performance occurs when both Cache Ratio ($CR$) and Miss Ratio ($MR$) are low. This scenario is typical in workloads with high temporal locality, where a small subset of "hot" data is frequently accessed. The cache tier can store this data, resulting in a low $CR$ while serving most requests (low $MR$).

\sstitle{Cost Disparity between Tiers:} 
The significant cost difference between cache and storage tiers is crucial for effective tiered storage. The cache tier offers high performance at a higher cost per unit capacity, while the storage tier provides lower performance at a much lower cost. This disparity allows the system to balance high performance for frequently accessed data with cost-effective storage for less accessed information.

\sstitle{Low Miss Penalty}: 
The miss penalty ($PC_{miss}$) represents the additional performance cost when a request misses the cache. A low miss penalty is vital for effective tiered storage as it reduces the impact of cache misses, allowing for a smaller cache (lower $CR$) without significantly degrading overall performance. We will present our techniques for minimizing $PC_{miss}$ in Subsection \ref{sec:writethough_back}.

\subsubsection{Workload-Driven Optimization Approaches}

Instead of immediately diving into TierBase-specific techniques, we can introduce a more general framework for mapping workload characteristics to optimization strategies. This sets the stage for the detailed techniques that will be discussed in the following section.

Table \ref{tab:workload_optimization} presents a mapping of various workload features to optimization strategies implemented in TierBase. This mapping illustrates how specific workload characteristics inform the choice of optimization techniques.



\begin{table}[]
\caption{Workload Features and Optimization Options}
\small
\begin{tabular}{c|c}
\hline
\textbf{Workload Features} & \textbf{Optimization Options} \\ \hline
\begin{tabular}[c]{@{}c@{}}Skewed access patterns \\  (a small subset of data accessed frequently)\end{tabular} & \begin{tabular}[c]{@{}c@{}}Tiered Storage\\ Elastic Threading\end{tabular} \\ \hline
Low latency requirements & \begin{tabular}[c]{@{}c@{}}In Memory Mode\\ PMem Usage\end{tabular} \\ \hline
\begin{tabular}[c]{@{}c@{}}Space-critical \\ (Large volume, low throughput)\end{tabular} & \begin{tabular}[c]{@{}c@{}}Larger Storage Instance\\ Tiered Storage\\ Pre-trained Compression\end{tabular} \\ \hline
\begin{tabular}[c]{@{}c@{}}Performance-critical \\ (High throughput, small volume)\end{tabular} & \begin{tabular}[c]{@{}c@{}}In Memory Mode\\ PMem for Persistence\end{tabular} \\ \hline
Read-heavy, Write-less & \begin{tabular}[c]{@{}c@{}}Elastic Threading\\ Pre-trained Compression\end{tabular} \\ \hline
Write-heavy & \begin{tabular}[c]{@{}c@{}}Write-back Caching\\ PMem for WAL\end{tabular} \\ \hline
\end{tabular}
\label{tab:workload_optimization}
\end{table}

These optimization techniques are designed to address specific workload characteristics and leverage the strengths of both single-tier and tiered storage architectures. By applying these strategies guided by our cost model, we can iteratively refine system configurations, resource allocations, and feature enablements.

In the following sections, we will detail each of these optimization techniques, explaining how they work and how they contribute to overall system cost-effectiveness. Later, in Section \ref{sec:exp}, we will demonstrate how these strategies are applied in practice. Our evaluation methodology involves replaying real-world workloads and assessing costs across various configurations, providing empirical validation of our cost-optimization approach.


\section{\tbase System Design}

\begin{figure*}[ht]
  \centering
  \includegraphics[width=0.8\textwidth]{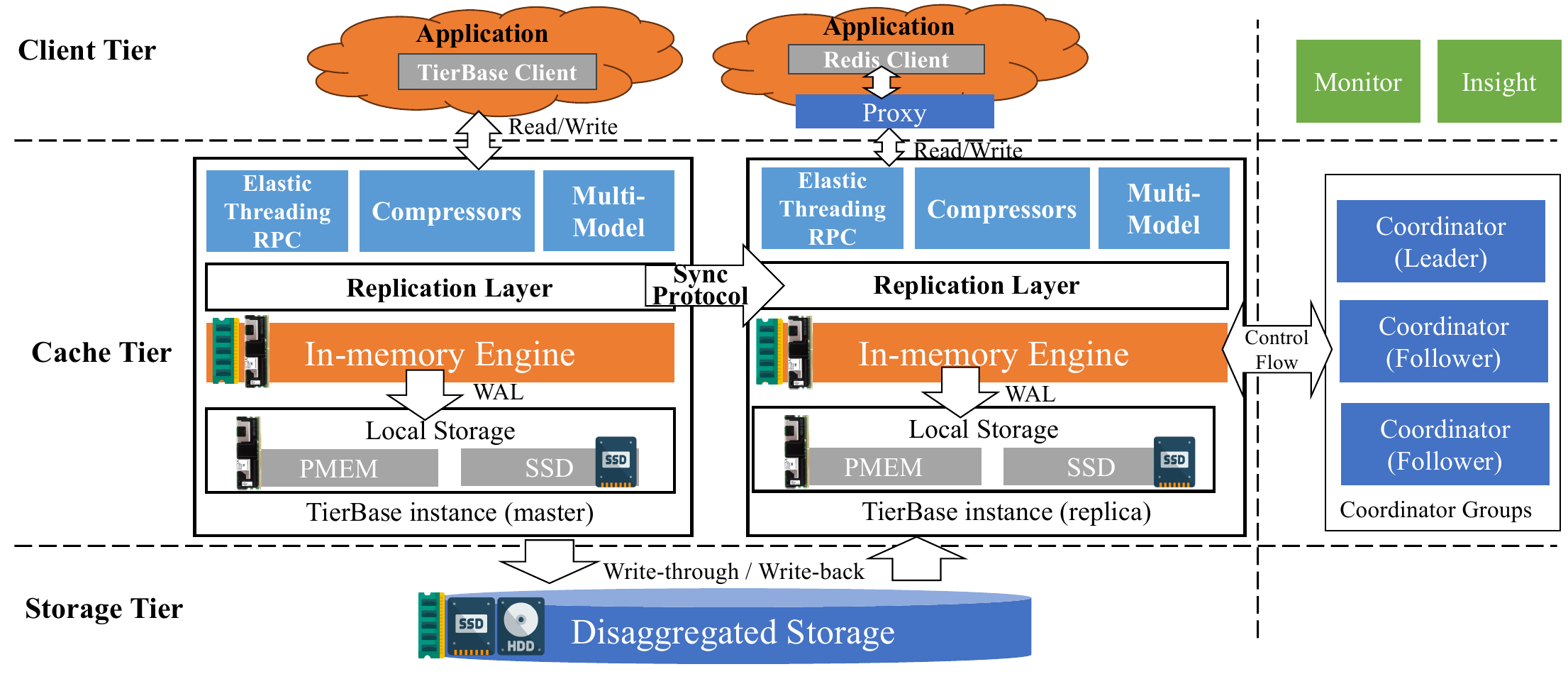}
  \caption{TierBase architecture overview}
  \label{fig:overview}
\end{figure*}

Building upon our unified space-performance cost model, we present TierBase, a high-performance, distributed key-value storage system designed to optimize cost for large-volume online storage. TierBase leverages a tiered storage architecture to provide low-latency, cost-effective data access while addressing the challenges highlighted in our cost model.

TierBase extends Redis's capabilities by supporting not only basic key-value operations like GET and SET, but also advanced data structures such as lists, sets, and sorted sets. Additionally, it provides CAS (Compare-And-Set) operations, wide-column data handling and vector search.

TierBase supports vector search by integrating the VSAG library\cite{vsag},  a vector indexing library developed by Ant Group for similarity search, which enables efficient ANN queries over high-dimensional vectors within our key-value infrastructure.  The integration supports dynamic vector operations, including real-time insertion and deletion in memory, demonstrating performance improvements of 3-4x compared to conventional algorithms such as HNSW.

As shown in Figure \ref{fig:overview}, TierBase incorporates a tiered storage architecture, separating the caching tier from the storage tier, allowing independent scaling based on workload requirements. 
The cache tier employs in-memory hash tables stored in DRAM or persistent memory (PMem) for efficient random access performance, while the storage tier typically utilizes a LSM-tree structure stored on SSD or HDD to optimize write performance and storage capacity. This architecture allows TierBase to effectively balance high-speed data access with efficient data storage across different storage media.

TierBase also features memory compression and elastic threading support, optimizing resource utilization dynamically.

The architecture of TierBase is structured into three primary tiers: client, cache, and storage.

\noindent\textbf{Client Tier}. The client tier consists of TierBase clients and proxy services. TierBase clients, compatible with native Redis clients, retrieve cluster routing information from the coordinator cluster for direct data access. They handle failover, cluster scaling automatically. For small-scale scenarios, TierBase provides a proxy service facilitating rapid integration for public cloud users.

\noindent\textbf{Cache Tier}. 
The cache tier in TierBase consists of instances and coordinators. Each instance serves as a data node with key hashes for data sharding. The cache instances implements hash tables for efficient key-value storage. Moreover, the cache tier can operate independently without the storage tier for in-memory storage use cases, providing high-speed data access similar to systems like Redis and Memcached. TierBase supports both single-replica and multi-replica modes, implementing various replication protocols to accommodate different reliability requirements. Coordinators oversee the entire cluster, managing failovers and administering tenant resource allocation.

\noindent\textbf{Storage Tier}. 
The storage tier provides data persistence through a disaggregated key-value storage system. The cache tier directly accesses this tier for cache misses, employing either write-through or write-back policies for data writing. 

TierBase offers various disaggregated storage options through a pluggable storage adapter. In our experiments, we focus on the Universal Configurable Storage (UCS)\footnote{UCS is an internal system developed by Ant Group and does not have a publicly available reference at the time of writing this paper.}, a sophisticated real-time serving and analytical storage engine. UCS implements an LSM-Tree with a shared disk architecture and remote compaction. This design ensures optimal online performance while supporting both row and columnar storage formats.

Although our experiments focus on an LSM-tree storage engine, the pluggable storage adapter in TierBase allows integration with various disaggregated storage systems based on different data structures. Consequently, the cost optimization techniques in the cache tier and cost evaluation with the cost model can be applied to a wide range of key-value stores.

Additionally, TierBase includes monitoring and analysis tools for real-time metrics collection, problem diagnosis, and workload-based suggestions. It integrates with Cougar\cite{Cougar} for automatic scaling in cloud environments.


\section{Cost Optimization Strategies}
In this section, we introduce the well-designed features which aim to optimize the cost of TierBase.
\label{Cost Optimization Strategies}

\subsection{Tiered Storage}
\label{sec:writethough_back}

TierBase introduces a tiered storage architecture that disaggregates cache and storage components, allowing for independent optimization. This approach directly addresses the space-performance trade-off highlighted in our cost model. The cache tier is optimized for speed (minimizing $PC_{cache}$), while the storage tier is designed for capacity and durability (optimizing $SC_{storage}$). Both tiers can scale independently, accommodating diverse workloads and data access patterns.

To ensure data consistency and reliability in this disaggregated architecture, we adapt the well-known caching techniques "write-through" and "write-back". These strategies are commonly used in conventional hardware scenarios to synchronize data between cache and storage. However, applying these techniques to a disaggregated architecture presents unique challenges. In traditional contexts, write-through and write-back policies are implemented within a single tier, where the cache and storage are tightly coupled. This tight coupling allows for simpler coordination and synchronization mechanisms between the cache and storage, as they reside on the same physical node or have low-latency communication channels. In contrast, the separation of cache and storage tiers in TierBase introduces new complexities in maintaining data consistency and synchronization, requiring careful design and implementation of coordination protocols and synchronization strategies.

\subsubsection{Write-through Caching} 
In the write-through caching policy (Figure \ref{fig:writethrough}), TierBase prioritizes data consistency between the cache and storage tiers. When a write request is received, it is first executed on the cache tier and then synchronously passed to the disaggregated storage tier before acknowledging completion to the application. If the storage update succeeds, the cache tier maintains the updated data; otherwise, the corresponding cache entry is invalidated, and an error is returned to the application.

To ensure data consistency in the presence of failures while reducing $PC_{miss}$, TierBase employs several key techniques:

\sstitle{Temporary Update Buffer.} Each connection maintains a temporary update buffer. Incoming update requests are initially performed on this buffer, and the results are used to update the main cache. If the storage write succeeds, the data is seamlessly transferred from the temporary buffer to the main cache. In case of a storage write failure, the corresponding entry in the main cache is removed, ensuring subsequent reads fetch the data from the storage, maintaining consistency.

\sstitle{Sequential Write Ordering.} TierBase uses a per-key write queue to maintain the sequential order of writes to the same key, ensuring consistent execution order in the cache tier for asynchronous storage updates.

\sstitle{Write Coalescing.}
Within Redis's event loop, TierBase coalesces multiple write commands targeting the same key into a single operation, similar to the concept of group commit in database systems. This approach efficiently updates the store with the final result, reducing the number of write operations and consequently lowering $PC_{miss}$.

This write-through strategy is particularly effective in environments where read operations significantly outnumber write operations, and high data reliability are critical. 

\begin{figure}[t]
   
\subfigure[Write-through caching policy]{
  \centering
  \includegraphics[scale=0.39]{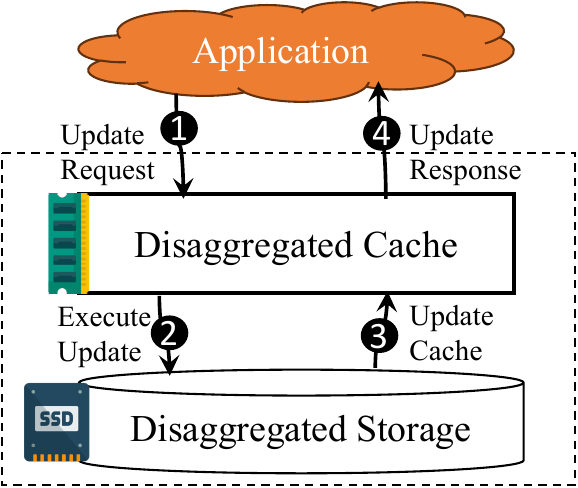}
  \label{fig:writethrough}
}
\hspace{-0.3cm}
\subfigure[Write-back caching policy]{
  \centering
  \includegraphics[scale=0.39]{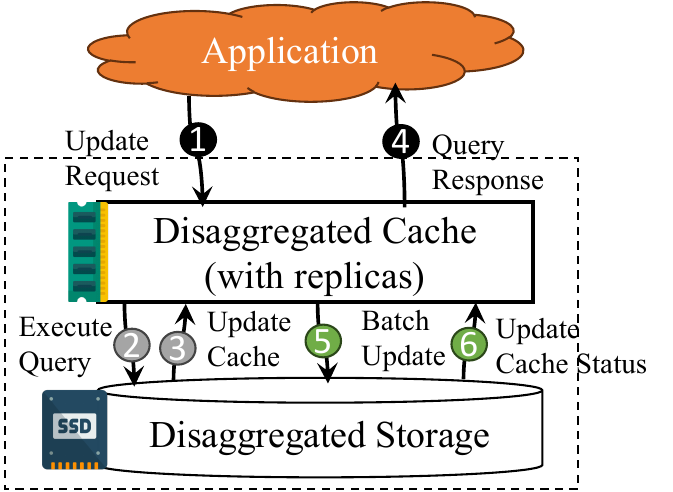}
  \label{fig:writeback}
}
\caption{Caching policy for write operation}
\label{fig:cachingpolicy}
\end{figure}

\subsubsection{Write-back Caching}

The write-back caching policy in TierBase prioritizes performance by optimizing $PC_{cache}$ and reducing $PC_{miss}$. Updates are first written to the cache tier with immediate response to the application, while data synchronization to the storage tier is deferred and performed asynchronously in batches. This approach minimizes $PC_{storage}$ by reducing the frequency of writes to the storage tier.

In cases where requested data is not present in the cache during an update operation, TierBase fetches the data from the storage tier before updating the cache. Data updated in the cache but not yet synchronized to storage is marked as``dirty" and periodically propagated in batches, minimizing remote calls to the storage tier.

Implementing write-back caching in a tiered storage architecture introduces unique challenges in ensuring data reliability and optimizing synchronization efficiency between cache and storage:

\sstitle{Replication of Cache.} To prevent data loss in case of cache tier failure, TierBase maintains multiple replicas of dirty data and cache contents.

\sstitle{Managing Dirty Data.} TierBase balances the scale of dirty data by restricting its size and establishing maximum interval times for batch updates. A backpressure mechanism is activated when dirty data approaches a predefined threshold.

\sstitle{Optimizing Update.} TierBase minimizes remote calls to the storage tier by batching updates and merging multiple updates for the same key.

\sstitle{Deferred Cache-fetching.} 
For update operations on missing keys, TierBase accumulates operations and submits batch read tasks to fetch data from storage, reducing read requests and minimizing costs in both tiers.

The consistency level in write-back caching is determined by the cache tier's configurable coherent protocol, ensuring updates are eventually propagated to underlying storage with strong consistency support.

\subsubsection{Write-through versus write-back}

TierBase supports both write-through and write-back caching policies, each offering different trade-offs:

Write-through caching provides lower $SC_{cache}$ as it doesn't store dirty data, but potentially higher $PC_{cache}$ for write-heavy workloads due to synchronous storage updates. Write-back caching offers lower $PC_{cache}$ and $PC_{miss}$ for write-heavy workloads due to deferred and batched storage updates, but incurs higher $SC_{cache}$ due to storing dirty data and potential replication.

The choice between these policies depends on specific workload characteristics and relative costs of cache and storage resources. Write-back caching may provide better cost-performance for write-heavy workloads with good temporal locality, while write-through caching may be more cost-effective for read-heavy workloads or when storage writes are relatively inexpensive.

TierBase's flexible configuration allows users to select the appropriate caching policy based on their application's requirements and cost structure, optimizing the space-performance trade-off for various workloads and scenarios.

\begin{figure}[tb]
\begin{center}
\begin{tabular}[t]{c}
 \includegraphics[scale=0.37]{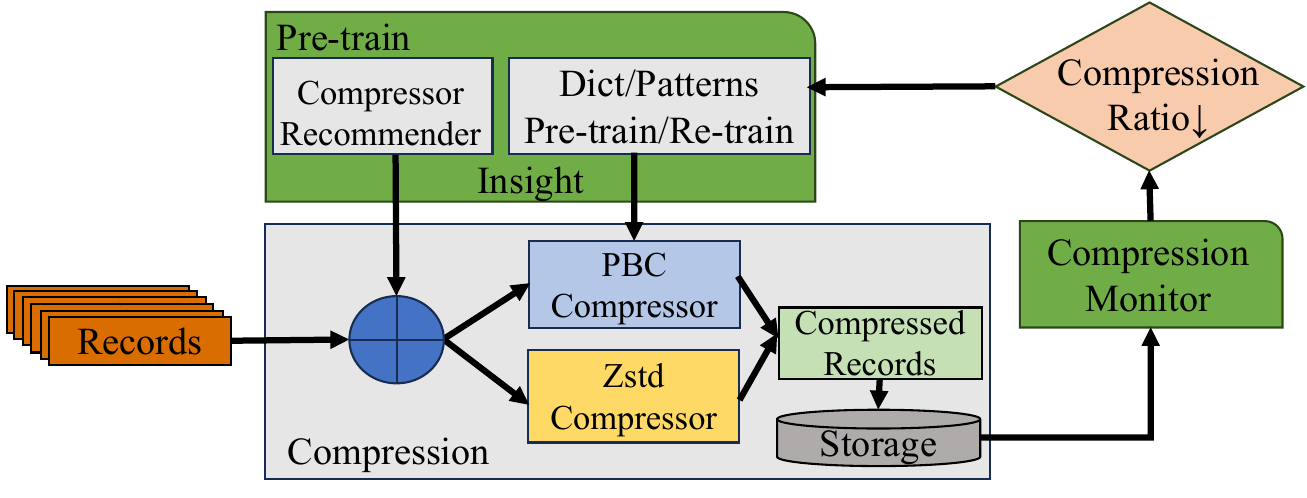}
\end{tabular}

\end{center}
\caption{The framework of pre-trained based compression}
\label{fig:compression-overview}
\end{figure}

\subsection{Pre-trained Compression Mechanism}
\label{pre_compress}

In the context of our space-performance cost model, in-memory data compression plays a crucial role in optimizing the trade-off between storage costs ($SC$) and performance costs ($PC$) within the memory tier. 

In TierBase, we develop a pre-trained compression mechanism which includes two efficient compression algorithms: our newly developed Pattern-Based Compression (PBC) \cite{pbc,github_pbc} and the widely adopted Zstandard (Zstd) \cite{zstd} by Meta. The framework of pre-trained compression is shown in Figure \ref{fig:compression-overview}. In the pre-training phase, Zstd builds a dictionary by identifying frequent strings in the data, while PBC employs hierarchical clustering and a unique similarity metric to pinpoint and extract data patterns. In the compression phase, the resulting patterns, along with residual strings, are then compressed further using string compression techniques.

Initially, we construct the dictionary (patterns) offline using samples from data records. We then apply this dictionary to the entire workload, enabling data compression and decompression.

A key challenge with pre-trained compression in production is the need to re-sample and re-train datasets when patterns change to avoid reduced compression ratios. To address this, TierBase introduces a monitoring service that continuously tracks compression efficiency and initiates re-sampling and re-training when necessary. Specifically, it monitors the compression ratio and the number of data records that do not align with the pattern. Re-sampling and retraining are triggered when the compression ratio falls below a baseline level or when the rate of unmatched records exceeds a predefined threshold.

Furthermore, TierBase’s Insight service includes a compressor recommender that automatically suggests the optimal compressor based on data types and performance requirements. This adaptive compression mechanism supports various data types, dynamically rebuilding the dictionary to accommodate changes in data patterns.

Our experiments demonstrate the cost-effectiveness of our pre-trained compression mechanism, enabling TierBase to dynamically adjust its compression strategy and balance the space-performance trade-off based on evolving data patterns and workloads. Despite a moderate increase in performance cost ($PC$) for write operations due to compression overhead, the significant reduction in space cost ($SC$) and high decompression speed for read operations optimize overall cost-effectiveness. By adjusting compression levels, TierBase can fine-tune the balance between $SC$ and $PC$, achieving an optimal point in the cost model to minimize total cost while maintaining high performance across diverse workloads.

\subsection{Persistent Memory Utilization}

A standout component in TierBase's storage strategy is the adoption of persistent memory(PMem).  The PMem distinguishes itself through its capacity, affordability compared to DRAM, swift memory-like access speeds, and its non-volatile nature.  The utilization of PMem enhances the overall performance while reduce the cost of TierBase in two folds:

\noindent\textbf{DRAM Extension}: Serving as an economical DRAM supplement, PMem allows for efficient memory use by keeping frequently accessed (hot) data, in DRAM, while less accessed (cold) data is stored in PMem. This strategy optimizes the balance between space cost ($SC$) and performance cost ($PC$).

\noindent\textbf{WAL Persistence}: PMem greatly improves TierBase’s Write-Ahead Log (WAL) persistence by overcoming the I/O operations per second (IOPS) bottleneck found in disk or cloud storage, ensuring faster and more consistent data synchronization. Crucially, WAL files are first written to a PMem-based persistent ring buffer, then batch-moved to cloud storage, achieving high throughput and real-time persistence, thus significantly boosting performance in high-demand scenarios.

To address the performance gap between PMem and DRAM, particularly in write latency, TierBase employs a refined memory allocation strategy. Small, frequently accessed data (keys and indexes) are stored in DRAM, while larger value data resides in PMem. Write operations to PMem are optimized through batching: data structures are assembled in DRAM before bulk transfer to PMem, reducing the impact on performance costs.

Our production experience demonstrates that PMem, when integrated into a tiered storage architecture alongside DRAM and SSDs, delivers strong performance even with straightforward data structures. This approach effectively balances performance and cost in TierBase's storage system.

\subsection{Elastic Threading}

\begin{figure}[t]
  \centering
  \includegraphics[scale=0.4]{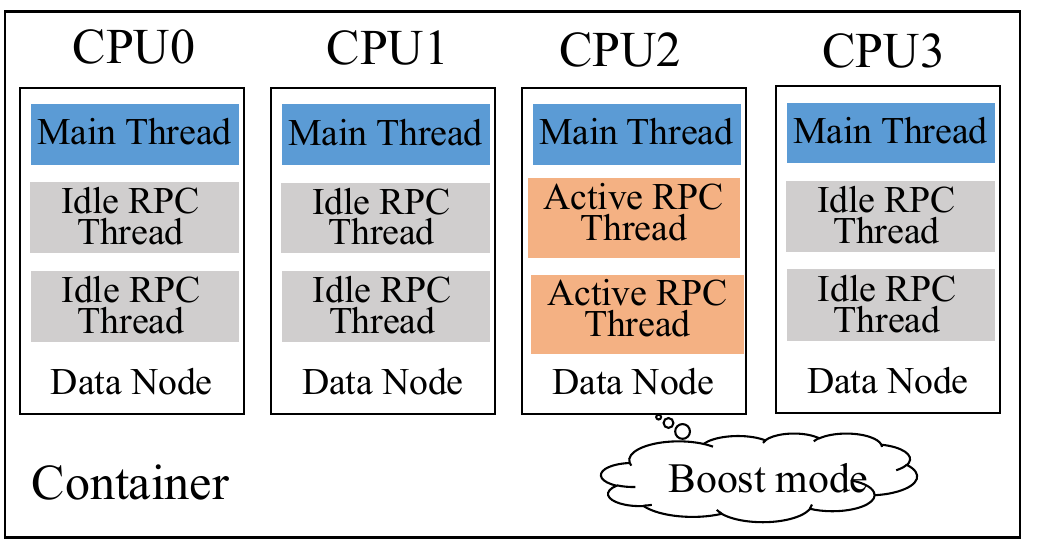}
  \caption{Elastic threading}
  \label{fig:elastic}
\end{figure}

TierBase implements an innovative elastic threading approach to optimize performance cost ($PC$) within allocated node resources. This dynamic method seamlessly switches between single-thread and multi-thread modes based on workload demands, enhancing system responsiveness and resource efficiency without external scaling.

In normal conditions, TierBase operates in a default single-thread mode, utilizing an event-driven model with epoll. This approach offers high CPU efficiency and lower $PC$ for typical workloads. We claim that the efficiency of a single-threaded process per data shard generally outperforms that of multi-threading due to reduced locking overhead, a principle supported by Amdahl's Law\cite{Amdahl}.

When the workload on a particular instance increases significantly, TierBase seamlessly transitions to multi-threaded mode by dynamically adding threads within the container's pre-allocated CPU resources. Containers are provisioned with CPU capacity based on anticipated peak workloads, but this capacity isn't fully utilized during normal operations. Elastic threading allows TierBase to leverage these underutilized CPU resources when needed, boosting the affected instance's performance without exceeding resource limits or incurring additional costs. When the workload subsides, TierBase switches back to single-thread mode, allowing CPU resources to be used by other processes within the container and maximizing resource efficiency.

Elastic threading is particularly effective for skewed workloads like dynamic hotspots. Typically, one instance might switch to multi-threaded mode while others remain in single-threaded mode within the same container, optimizing resources across the system. If the container's overall CPU load remains consistently high, the system recognizes the need to scale out to further enhance tenant performance.

This approach improves responsiveness to immediate demands and optimizes resource use, preventing unnecessary allocation during low-activity periods. By dynamically balancing between single-threaded efficiency and multi-threaded performance, elastic threading significantly contributes to overall CPU  efficiency. Elastic threading uses only idle CPU resources within the allocated container. Threads are dynamically added or removed based on workload demands, ensuring efficient resource use without over-provisioning or incurring extra costs.

\section{Further Cost Analysis and Framework}

\label{detailed_analysis}

This section provides an in-depth analysis of our cost model by examining space-performance trade-offs in storage systems, relating these to our Optimal Cost Theorem and the classic Five-Minute Rule. We adapt the Five-Minute Rule for modern distributed systems and provide a framework for cost optimization for tiered storage.

\subsection{Adapting the Five-Minute Rule for Modern Storage Systems}
\label{interval}
The Five-Minute Rule\cite{gray19875,five20,five30}, introduced by Jim Gray and Gianfranco Putzolu in 1985, has been a cornerstone in database system design. Originally formulated for single-server environments, it provided a simple heuristic for deciding whether data should be kept in memory or on disk based on its access frequency:

{\small
\setlength{\abovedisplayskip}{5pt}
\setlength{\belowdisplayskip}{5pt}
\begin{align}
BreakEvenInterval &= \left( \frac{PagesPerMBofRAM}{AccessPerSecondPerDisk} \right) \nonumber \\
&\quad \times \left( \frac{PricePerDiskDrive}{PricePerMBofRAM} \right) \label{eq:rule_equation}
\end{align}
}

However, in today's distributed and cloud-based systems, we need to consider a broader range of factors and trade-offs. We propose an adapted version of the Five-Minute Rule that aligns with our cost model:
{\small
\begin{equation}
BreakEvenInterval = \frac{CPQPS_{slow}}{CPGB_{fast} \times AverageRecordSize}
\end{equation}
}

Where $CPQPS_{slow}$ is the Cost Per Query Per Second for slower, space-optimized storage, $CPGB_{fast}$ is the Cost Per Gigabyte for faster, performance-optimized storage, and $AverageRecordSize$ is the average size of data records in the workload.

To illustrate how our adapted rule relates to the original, we provide the following mapping.
\begin{itemize}
    \item The ratio $\left( \frac{PricePerDiskDrive}{AccessPerSecondPerDisk} \right)$ effectively represents the $CPQPS_{slow}$.
    \item $PricePerMBofRAM$ correlates to $CPGB_{fast}$.
    \item $PagesPerMBofRAM$ is conceptually similar to \\$\left( \frac{1GB}{AverageRecordSize} \right)$.
\end{itemize}

This formulation determines the optimal point on the space-performance trade-off spectrum for a given workload in a distributed environment. It can be illustrated by comparing fast, in-memory storage systems like Redis \cite{redis} with slower, more space-efficient systems like HBase\cite{hadoop}, or by examining different configurations of the same database system.

The break-even interval analysis, derived from our adapted Five-Minute Rule, optimizes data placement by comparing data access intervals to the break-even point. It guides the choice between fast, performance-oriented storage and slower, space-efficient options.

Our Cost Optimal Theorem extends this concept by determining the best overall storage settings for a given workload, aiming to minimize system cost by balancing performance cost ($PC$) and space cost ($SC$). It considers the entire system configuration, including multiple storage tiers and complex workload characteristics, to guide high-level design and resource allocation decisions.

Our comprehensive experiments assessed the performance and cost efficiency of our system against leading open-source key-value databases using this integrated approach. The evaluation results on real application workloads demonstrate the model's effectiveness in guiding cost-saving database system designs. The case study in Section \ref{case_study} showcases how the break-even interval, derived from the Five-Minute Rule, helps choose the most cost-effective TierBase configuration within the framework established by the Cost Optimal Theorem.

\subsection{Cost Analysis of Tiered Storage}
\label{cost_analysis_tired_storage}

Revisiting Equation \ref{eq:tiered_cost}, we can focus on optimizing the cost of the cache tier. In disaggregated storage systems with a sufficiently large storage pool, the storage tier cost is dominated by $SC$ when $MR < SC_{storage} / PC_{storage}$ for skewed access patterns. As illustrated in Figure \ref{fig:cost-model2}, this allows us to concentrate on the cache tier cost:

\begin{align}
Cost_{cache} = \max(&PC_{cache} + PC_{miss} \times MR, \nonumber \\
&SC_{cache} \times CR)
\end{align}

To find the optimal cost, we consider the relationship between the Miss Ratio ($MR$) and the Cache Ratio ($CR$), typically represented by the Miss Ratio Curve\cite{10.1145/3185751}, where $MR = f(CR)$, and $f$ is a non-increasing function.

\begin{theorem}[Optimal Cache Tier Cost]
The optimal cost for the cache tier of a tiered storage is achieved when the performance cost equals the space cost:
\begin{equation}
PC_{cache} + PC_{miss} \times f(CR^*) = SC_{cache} \times CR^*
\end{equation}
where $CR^*$ is the optimal cache ratio.
\end{theorem}

Let $CR^*$ be the optimal cache ratio that minimizes the overall cost of the cache tier:

$Cost_{cache}^* = \min_{0 \leq CR \leq 1} \max(PC_{cache} + PC_{miss} \times f(CR), SC_{cache} \times CR)$

Define two functions:

$g(CR) = PC_{cache} + PC_{miss} \times f(CR)$
$h(CR) = SC_{cache} \times CR$

Note that $g(CR)$ is non-increasing (as $f(CR)$ is non-increasing) and $h(CR)$ is increasing linearly with $CR$.

The optimal cost occurs at the intersection of these two functions. To see why, consider:

1. If $g(CR) > h(CR)$, we can decrease $CR$ to reduce cost.
2. If $g(CR) < h(CR)$, we can increase $CR$ to reduce cost.
3. The minimum cost occurs when neither of these improvements is possible, i.e., at $g(CR) = h(CR)$.

Therefore, the optimal cache ratio $CR^*$ satisfies:
\begin{equation}
PC_{cache} + PC_{miss} \times f(CR^*) = SC_{cache} \times CR^*
\end{equation}

This equality represents the balance point where performance cost (including miss penalty) equals space cost, minimizing the overall cache tier cost.
This theorem provides a principle for optimizing tiered storage systems: the most cost-effective configuration is one where the cache tier's performance cost (including the cost of cache misses) equals its space cost. This balance point represents the optimal trade-off between performance and space for the cache tier.

In practice, estimating the exact $CR^*$ is challenging, as $f(CR)$ can be complex and highly dependent on specific workload characteristics. Nevertheless, this theorem serves as a valuable target for optimization efforts, guiding cache size tuning. To address this challenge, we propose an evaluation-based approach in Section \ref{sec:framework} to find the optimal $CR$.

Furthermore, this analysis provides a way for determining when to use tiered storage over single-tier solutions and how to optimally configure the cache tier within a tiered storage system, enhancing our understanding of cost-effective storage design in modern disaggregated environments.

\subsection{Cost Optimization Framework}
\label{sec:framework}

In order to speedup the cost optimization procedure, we develop a sample-based method to calculate the cost for various configurations with regard to real-world workloads. The method involves the following steps:

\begin{enumerate}
\item \textbf{Sample}: Sample data snapshots and record a representative period of workload from production instances.
\item \textbf{Load}: Load the sampled data snapshot into a testing instance with a specific configuration.
\item \textbf{Replay}: Replay the recorded real-world key-value operation traces on the testing instance, measuring and collecting the maximum performance and maximum space utilization for the workload.
\item \textbf{Calculation}: Calculate the workload cost based on measurements.
\item \textbf{Iteration}: Repeatedly perform steps 2-4 with different configurations to approach cost-optimal configuration.
\end{enumerate}

This method simulates key-value store behavior under realistic conditions, providing accurate performance and cost assessments. By using real workload traces and access patterns, we obtain a precise representation of system performance. This method enables comprehensive exploration of the configuration space, ensuring identification of the most cost-effective configuration for each individual workload.

While the configuration space for cost optimization can be large, in practice, we focus on the most impactful parameters specific to the workload, guided by user input and prior experience. This approach narrows the candidate configurations substantially. Optimization computations are performed offline and parallelized to accelerate the process, ensuring that the time invested is minor compared to the long-term cost savings achieved. To address the cold start problem, we initialize the system with configurations based on user inputs and best practices.The detailed optimization results and cost savings are analyzed with case studies in Section \ref{sec:exp}.

\section{Experiments}
\label{sec:exp}

\subsection{Settings}
The experimental evaluation is conducted on the following servers. For the cache tier, we use three servers with dual Intel Xeon Platinum 8263C CPUs at 2.50 GHz, 192GB DRAM, and eight 128GB Intel\textsuperscript{\textregistered} Optane\texttrademark{} DCPMM 100 series(App Direct Mode), while for the storage tier, we use three servers with Intel Xeon Platinum 8163 CPU at 2.50 GHz, 64GB DRAM, and 8TB SSDs. For all servers, NUMA is enabled and Hyper-threading is disabled for linear performance scaling and to avoid contention between logical cores. The system and software environments are as follows: Linux kernel version 4.19.91, OpenJDK version 1.8.0, GCC version 10.2.1, Dragonfly 1.23.0, Redis 6.0.17, Cassandra 4.0.11, HBase 2.4.1, and Memcached 1.6.20.

YCSB (Yahoo! Cloud Serving Benchmark) \cite{ycsb} is utilized which encompasses the load and run phase. Our experiments utilize two distinct default workloads from YCSB: Workload A, characterized by a predominance of write operations, and Workload B, distinguished by a higher proportion of read operations. We have adapted YCSB to accept user-specified datasets for data insertion, as opposed to the default use of random strings as values. In particular, the Cities dataset \cite{cities} is designated as the default for our tests. We deploy 16 YCSB threads for single-thread cases and 48 for multi-thread cases.

We focus our comparisons on widely-used production systems to evaluate TierBase's performance and cost-effectiveness in real deployment scenarios, which have essential features implemented, such as full failure recovery. These systems can reflect the performance and cost accurately in real world applications.
We selected Redis \cite{redis}, Memcached \cite{memcached} and Dragonfly \cite{dragonflydb} for  caching system comparison. Redis and Memcached are established caching systems, extensively utilized across diverse applications. Dragonfly is a newly introduced, high-performance caching system. For databases with persistence, we select Redis with AOF, Cassandra \cite{Cassandra}, and HBase \cite{hadoop} as competitors. Specifically, Redis-AOF ensures data durability by logging writes to disk, which can impact performance due to the additional disk I/O overhead.

Instances represent the fundamental units of resource allocation. In systems operating in single-thread mode, each instance is allocated 1 CPU core and 4GB of memory. In multi-thread mode and databases with persistence, the allocation for each instance increases to 4 CPU cores and 16GB of memory. These specifications are common instance configurations used by Ant Group.

\subsection{Performance Evaluation}
\subsubsection{Caching Systems} 

Figure \ref{fig:performance} illustrates the performance comparison of four caching systems: TierBase, Redis, Memcached, and Dragonfly. During this evaluation, TierBase was tested in its default mode without any cost optimization techniques enabled. We evaluate the performance for single-thread and multi-thread mode separately and report the throughput and 99th percentile tail latency respectively.

In single-thread mode (Figures \ref{fig:performance}(a) and \ref{fig:performance}(b)), TierBase and Redis exhibit similar performance, outperforming Memcached and Dragonfly across all workloads. This distinction arises because Memcached and Dragonfly are principally engineered for multi-thread environments, while in contrast, Redis is meticulously optimized for single-thread mode. TierBase maintains the lowest latency across most workloads. During the load phase, the latency of TierBase and Redis is significantly lower than that of Memcached and Dragonfly.

In multi-thread mode (Figures \ref{fig:performance}(c) and \ref{fig:performance}(d)), Memcached and Dragonfly surpass TierBase and Redis. Memcached combines a streamlined caching approach with a lightweight threading model, minimizing inter-thread contention and enhancing speed. Dragonfly benefits from a shared-nothing architecture for threads, boosting its parallel processing. Although TierBase's per-instance throughput is slightly lower in multi-thread mode, it excels in real-world scenarios through efficient scaling. Figure \ref{fig:performance}(c) shows that 4 single-threaded TierBase instances outperform a single multi-threaded instance of Memcached or Dragonfly using equivalent resources, leading to a lower performance cost. In practice, scaling out across multiple instances meets performance requirements, and TierBase's cost-effectiveness makes this economically viable. Thus, TierBase effectively balances performance and cost, aligning with the demands of large-scale applications.

\input{figures/performance_multi}
\input{figures/persistence_single}%
\setlength{\textfloatsep}{10pt}
\subsubsection{Persistence Mechanisms} 
Figure \ref{fig:ps} outlines TierBase's performance with four persistence mechanisms in single-thread mode: WAL, WAL while using PMem as the persistent ring buffer(WAL-PMem), write-back and write-through introduced in Section \ref{sec:writethough_back}.

For the throughput performance(Figure \ref{fig:ps}(a)), write-back significantly outperforms write-through in the load phase by 92.52\%, due to its deferred writing mechanism that reduces immediate write operation overhead. In various read-write workloads, write-back's throughput is about twice of write-through, demonstrating its efficiency in write-intensive tasks. The WAL-PMem mode, while not as effective as write-back, still surpasses write-through, indicating benefits by using persistent memory. 
However, WAL mode outperforms WAL-PMem due to its use of SSDs and asynchronous disk flushes every second, while WAL-PMem synchronizes to PMem per transaction, potentially incurring higher synchronization overhead.

In terms of latency (Figure \ref{fig:ps}(b)), write-through experiences the highest latency due to its immediate write to storage, being around 3 times higher than write-back in the load phase. Write-back, with its deferred write, significantly lowers latency, particularly in write-heavy scenarios. WAL-PMem offers a middle ground, with lower latency than write-through but higher than write-back.

\subsection{Features Evaluation}

\begin{table}[]
\centering
 \caption{Evaluation of compression techniques}
\begin{tabular}{cc|c|c|c}
\hline
\multicolumn{2}{c|}{Datasets} & Cities & \multicolumn{1}{l|}{KV1} & \multicolumn{1}{l}{KV2} \\ \hline
\multicolumn{1}{c|}{\multirow{2}{*}{Comp. Ratio}} & PBC & \textbf{0.2003} & \textbf{0.2341} & \textbf{0.2297} \\
\multicolumn{1}{c|}{} & Zstd & 0.2920 & 0.4042 & 0.4096 \\ \hline
\multicolumn{1}{c|}{\multirow{2}{*}{\begin{tabular}[c]{@{}c@{}}Overall \\ Comp. Ratio\end{tabular}}} & PBC & \textbf{0.4919} & \textbf{0.6884} & \textbf{0.6219} \\
\multicolumn{1}{c|}{} & Zstd & 0.5508 & 0.7594 & 0.7117 \\ \hline
\multicolumn{1}{c|}{\multirow{3}{*}{\begin{tabular}[c]{@{}c@{}}Throughput\\ (SET)\end{tabular}}} & PBC & 54469 & 74878 & 65329 \\
\multicolumn{1}{c|}{} & Zstd & 61667 & 75018 & 67558 \\
\multicolumn{1}{c|}{} & Raw & \textbf{122324} & \textbf{122414} & \textbf{120496} \\ \hline
\multicolumn{1}{c|}{\multirow{3}{*}{\begin{tabular}[c]{@{}c@{}}Throughput\\ (GET)\end{tabular}}} & PBC & 109998 & 119688 & 115888 \\
\multicolumn{1}{c|}{} & Zstd & 96861 & 112866 & 102616 \\
\multicolumn{1}{c|}{} & Raw & \textbf{138045} & \textbf{132397} & \textbf{132961} \\ \hline
\end{tabular}
 
 \label{fig:compress study}
\end{table}

\subsubsection{Compression.}
As introduced in \ref{pre_compress}, \tbase has implemented pre-trained based compression strategies to mitigate memory utilization. We evaluate effectivness of the pre-trained based compression methods:(i) Zstd-b, (ii) Zstd-d, (iii) PBC. Basic Zstd\cite{zstd} (denoted as Zstd-b) which is without pre-trained dictionaries, Zstd with pre-trained dictionaries(denoted as Zstd-d) and Pattern-Based Compression\cite{pbc} (PBC). We also include the raw data without compression(Raw) as the bar of throughput.

As shown in Table \ref{fig:compress study}, the pre-trained based methods, PBC and Zstd-d, consistently outperform Zstd-b. It demonstrates that the pre-trained mechanism, by prior analysis and storage of common data patterns, is able to enhance the compression ratio. Notably, PBC consistently achieves higher compression ratios than Zstd.  In KV datasets, the distinctive patterns within the values lead to a more significant improvement in PBC's compression performance.  Specifically, PBC surpasses Zstd-d by 43\% and Zstd-b by 74\% in average compression ratios. These enhanced ratios contribute to PBC's sustained superiority in overall compression performance.

In the evaluation of throughput. all three compression methods perform worse compared to Raw, especially in SET operations. Among the three compression mechanisms tested, Zstd-d demonstrated the highest performance. In public datasets, the throughput of Zstd-d was approximately twice as high as PBC and about 3.5 times higher than Zstd-b. This notable difference is primarily due to the higher computational overhead in PBC's compression process which involves pattern matching and string encoding. Meanwhile, Zstd-b, without pre-trained dictionary, necessitates online data analysis during compression, which hampers throughput. In contrast, when considering average GET operation throughput, PBC not only surpasses Zstd-d but also nearly parallels the velocity of Raw. Zstd-b still demonstrates the least favorable performance.

The implementation of compression introduces a calculated trade-off, modestly impeding throughput performance in exchange for substantial memory conservation, thereby augmenting the judicious utilization of resources. Based on the pre-trained compression strategy, it is possible to preserve common patterns in the data in advance, avoiding real-time computation during compression and decompression, thereby enhancing the efficiency and effectiveness of the compression process. The Section \ref{sec:cost evaluation} will provide an in-depth exploration of the multifaceted benefits engendered by compression.

\subsubsection{Elastic threading.}
\indent
In order to show the effect of elastic threading,  we simulate a scenario of workload burst to test the system's adaptability to sudden workload influxes. At the beginning, the workload maintenance a low QPS(20,000). Then, we increase in the number of client requests to simulate a surge in workload at the 15 second. This state lasts for 30 seconds and finally the workload returned to normal, the low QPS state. We represent single-thread, multi-thread and elastic threading modes as $s$, $m$ and $e$ (e.g. TierBase-s, TierBase-m, TierBase-e).

As shown in Figure \ref{fig:elastic-thread}, under normal conditions, all databases manage well, with Redis showing some jitter in multi-thread mode. Upon increased workload, systems' throughput reaches their limits, and latency rise. In single-thread mode, \tbase has the highest latency, followed by Redis. However, with elastic threading, \tbase initially faces higher latency but quickly adjusts to have the lowest latency, equal to its multi-thread mode performance. In terms of throughput, both \tbase and Redis hit 120,000 QPS in single-thread mode, with Redis peaking at 180,000 QPS and \tbase at over 240,000 QPS in multi-thread mode. 

In summary, the experiment result indicates that elastic threading allows \tbase to operate in a cost-saving single-thread mode under normal scenario, while automatically switching to a multi-thread mode to achieve higher throughput performance during workload spikes, without the need for manual intervention.
\begin{figure}[t]
\centering
\begin{tikzpicture}[scale=0.6]
\hspace{0.1cm}

\begin{axis}[
    hide axis,
    scale only axis,
    height=0pt,
    width=100pt,
    legend style={anchor=north west,legend columns=5,font=\Large,fill opacity=0.7, draw opacity=1, row sep=-4pt,draw=black, thick,line width=1pt},
    legend image post style={line width=2pt, scale=1},
    xmin=0, xmax=1, ymin=0, ymax=1,
]
\addlegendimage{color=c4,mark=none,line width=2pt}
\addlegendentry{TierBase-s}
\addlegendimage{color=c5,mark=none,line width=2pt}
\addlegendentry{TierBase-e}
\addlegendimage{color=c8,mark=none,line width=2pt}
\addlegendentry{TierBase-m}
\addlegendimage{color=c1,mark=none,line width=2pt}
\addlegendentry{Redis-s}
\addlegendimage{color=c3,mark=none,line width=2pt}
\addlegendentry{Redis-m}

\end{axis}
\end{tikzpicture}
\begin{tikzpicture}[scale=0.42]
\begin{axis}[
    width=1\textwidth,
    height=0.37\textwidth,
    xlabel={\textbf{\Huge time(s)}},
    ylabel={\textbf{\Huge throughput ($k$ qps)}},
    legend pos=north west, 
    legend style={at={(0.13,0.95)},legend columns=4,font=\huge,fill opacity=0.7, draw opacity=1, row sep=-4pt,draw=black, thick,},
    legend image post style={line width=2pt, scale=1},
    xlabel near ticks, 
    ylabel near ticks, 
    scaled y ticks=false,
    tickwidth=3pt,
    axis line style={line width=2pt}, 
    xtick={0,10,20,30,40,50,60},
    ytick={50000,100000,150000,200000,250000}, 
    yticklabels={50,100,150,200,250},
    tick label style={font=\huge},
    xmin=0, 
    ymin=0, 
    xmax=60,
    ymax=250000,
]

\addplot[color=c4,mark=none,line width=2.5pt] coordinates {
    (0.0,19488.0)(1.0,19488.0)(2.0,20000.0)(3.0,20000.0)(4.0,20000.0)(5.0,20000.0)(6.0,20000.0)(7.0,19999.0)(8.0,20001.0)(9.0,20000.0)(10.0,20000.0)(11.0,20000.0)(12.0,20000.0)(13.0,20000.0)(14.0,20000.0)(15.0,20000.0)(16.0,117079.0)(17.0,125656.0)(18.0,127974.0)(19.0,127606.0)(20.0,126807.0)(21.0,127185.0)(22.0,127744.0)(23.0,127534.0)(24.0,127745.0)(25.0,127004.0)(26.0,127434.0)(27.0,127843.0)(28.0,127844.0)(29.0,127726.0)(30.0,127855.0)(31.0,127595.0)(32.0,127419.0)(33.0,127957.0)(34.0,127515.0)(35.0,128424.0)(36.0,127045.0)(37.0,127111.0)(38.0,127226.0)(39.0,127681.0)(40.0,125741.0)(41.0,127648.0)(42.0,127846.0)(43.0,127109.0)(44.0,127494.0)(45.0,127964.0)(46.0,20004.0)(47.0,20000.0)(48.0,20000.0)(49.0,20000.0)(50.0,20000.0)(51.0,20000.0)(52.0,20000.0)(53.0,20000.0)(54.0,20000.0)(55.0,20000.0)(56.0,20000.0)(57.0,20000.0)(58.0,20000.0)(59.0,20000.0)(60.0,20000.0)
};

\addplot[color=c8,mark=none,line width=2.5pt] coordinates {
    (0.0,19488.0)(1.0,19536.0)(2.0,19999.0)(3.0,20001.0)(4.0,20000.0)(5.0,20000.0)(6.0,20000.0)(7.0,20000.0)(8.0,19999.0)(9.0,20001.0)(10.0,20000.0)(11.0,20000.0)(12.0,20000.0)(13.0,20000.0)(14.0,20000.0)(15.0,20000.0)(16.0,201025.0)(17.0,230721.0)(18.0,231544.0)(19.0,231982.0)(20.0,231885.0)(21.0,232184.0)(22.0,232251.0)(23.0,231499.0)(24.0,231904.0)(25.0,231440.0)(26.0,231700.0)(27.0,231986.0)(28.0,233104.0)(29.0,231452.0)(30.0,231883.0)(31.0,232374.0)(32.0,232605.0)(33.0,232432.0)(34.0,224010.0)(35.0,232374.0)(36.0,232387.0)(37.0,231040.0)(38.0,230614.0)(39.0,231530.0)(40.0,231595.0)(41.0,230363.0)(42.0,232270.0)(43.0,232503.0)(44.0,232173.0)(45.0,232641.0)(46.0,20000.0)(47.0,19999.0)(48.0,20001.0)(49.0,20000.0)(50.0,20000.0)(51.0,20000.0)(52.0,20000.0)(53.0,20000.0)(54.0,20000.0)(55.0,20000.0)(56.0,20000.0)(57.0,20000.0)(58.0,19999.0)(59.0,20001.0)(60.0,20000.0)
};

\addplot[color=c1,mark=none,line width=2.5pt] coordinates {
    (0.0,19536.0)(1.0,19536.0)(2.0,20000.0)(3.0,20000.0)(4.0,20000.0)(5.0,20000.0)(6.0,20000.0)(7.0,20000.0)(8.0,20000.0)(9.0,20000.0)(10.0,20000.0)(11.0,20000.0)(12.0,20000.0)(13.0,20000.0)(14.0,20000.0)(15.0,20000.0)(16.0,120388.0)(17.0,130602.0)(18.0,132450.0)(19.0,132997.0)(20.0,132966.0)(21.0,133058.0)(22.0,133003.0)(23.0,132950.0)(24.0,133370.0)(25.0,133042.0)(26.0,133252.0)(27.0,134094.0)(28.0,134062.0)(29.0,133764.0)(30.0,133300.0)(31.0,133865.0)(32.0,133665.0)(33.0,133749.0)(34.0,133206.0)(35.0,133259.0)(36.0,133309.0)(37.0,133216.0)(38.0,133887.0)(39.0,133845.0)(40.0,133726.0)(41.0,133714.0)(42.0,132199.0)(43.0,133033.0)(44.0,133242.0)(45.0,132800.0)(46.0,20001.0)(47.0,20000.0)(48.0,20000.0)(49.0,20000.0)(50.0,20000.0)(51.0,20000.0)(52.0,20000.0)(53.0,20000.0)(54.0,20000.0)(55.0,20000.0)(56.0,20000.0)(57.0,20000.0)(58.0,20000.0)(59.0,20000.0)(60.0,20000.0)
};

\addplot[color=c5,mark=none,line width=2.5pt] coordinates {
        (0.0,19516.48)(1.0,19516.48)(2.0,20000.0)(3.0,20000.0)(4.0,20000.0)(5.0,20000.0)(6.0,20000.0)(7.0,20000.0)(8.0,20000.0)(9.0,20000.0)(10.0,20000.0)(11.0,20000.0)(12.0,20000.0)(13.0,20000.0)(14.0,20000.0)(15.0,20000.0)(16.0,119579.0)(17.0,129090.0)(18.0,131563.0)(19.0,158616.0)(20.0,236259.0)(21.0,237313.0)(22.0,237910.0)(23.0,236167.0)(24.0,236718.0)(25.0,235826.0)(26.0,236400.0)(27.0,235382.0)(28.0,236266.0)(29.0,235692.0)(30.0,235949.0)(31.0,235648.0)(32.0,237195.0)(33.0,236669.0)(34.0,237532.0)(35.0,236905.0)(36.0,235698.0)(37.0,236585.0)(38.0,236028.0)(39.0,236046.0)(40.0,237090.0)(41.0,235845.0)(42.0,236865.0)(43.0,234758.0)(44.0,236447.0)(45.0,236306.0)(46.0,20000.0)(47.0,20000.0)(48.0,20000.0)(49.0,20000.0)(50.0,20000.0)(51.0,20000.0)(52.0,20000.0)(53.0,20000.0)(54.0,20000.0)(55.0,20000.0)(56.0,20000.0)(57.0,20000.0)(58.0,20000.0)(59.0,20000.0)(60.0,20000.0)
};

\addplot[color=c3,mark=none,line width=2.5pt] coordinates {
    (0.0,19515.48)(1.0,19515.48)(2.0,19985.0)(3.0,20001.0)(4.0,20019.02)(5.0,20000.0)(6.0,19980.02)(7.0,19649.65)(8.0,20370.0)(9.0,20000.0)(10.0,20000.0)(11.0,19980.02)(12.0,20001.0)(13.0,20004.0)(14.0,19995.0)(15.0,20000.0)(16.0,154378.0)(17.0,172106.0)(18.0,175443.0)(19.0,175994.0)(20.0,177979.0)(21.0,178280.0)(22.0,176984.0)(23.0,177939.0)(24.0,177218.0)(25.0,178613.0)(26.0,179861.0)(27.0,179722.0)(28.0,180756.0)(29.0,184383.0)(30.0,184638.0)(31.0,185038.0)(32.0,184709.0)(33.0,185969.0)(34.0,185273.0)(35.0,184723.0)(36.0,184724.0)(37.0,185282.0)(38.0,185161.0)(39.0,185137.0)(40.0,185131.0)(41.0,185022.0)(42.0,185338.0)(43.0,184819.0)(44.0,185208.0)(45.0,185466.0)(46.0,20000.0)(47.0,20000.0)(48.0,20000.0)(49.0,20000.0)(50.0,20000.0)(51.0,20000.0)(52.0,20000.0)(53.0,20000.0)(54.0,20000.0)(55.0,20000.0)(56.0,18776.0)(57.0,21224.0)(58.0,20000.0)(59.0,20000.0)(60.0,20020.02)
};

\end{axis}
\end{tikzpicture}
\caption{Performance of TierBase and Redis for workload boosting}
\label{fig:elastic-thread}
\end{figure}
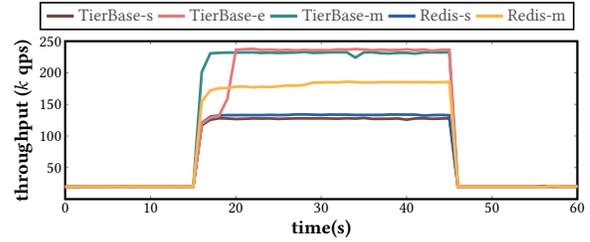

\subsection{Cost Evaluation} \label{sec:cost evaluation}

In the following two subsections, we evaluate the cost-effectiveness of TierBase and compare it with other representative systems under different configurations and synthetic workloads using the space-performance cost model. 

\subsubsection{Evaluation setup}

We employ the framework introduced in Section \ref{sec:framework} to evaluate synthetic workloads generated by YCSB using public datasets for write operations. Our simulated workload comprises 10GB data with 80,000 QPS for caching systems, and 10GB data with 40,000 QPS for databases with persistence. While our cost model is applicable to various workloads, we selected these specific parameters as a representative miniature of typical workloads at Ant Group.

The cost unit presented is relative, based on a standard container with 1 CPU core and 4GB of memory. All systems are tested within this standard container on a single instance, with $CPQPS$ and $CPGB$ calculated accordingly.
This evaluation methodology allows us to assess the cost-effectiveness of various systems under controlled conditions.

For systems employing replication (e.g., Redis with AOF, TierBase with WAL, and TierBase with write-back policy), we implement a master-replica setup in the cache tier to ensure data reliability. This configuration effectively doubles the cache tier cost.

We denote single-thread, multi-thread, and elastic threading as $s$, $m$, and $e$ respectively. TierBase-PMem is for PMem activation in TierBase, and TierBase-Zstd and TierBase-PBC for compression. Redis-AOF and TierBase-WAL denote Redis with AOF and TierBase with WAL, respectively. Write-through and write-back policies in TierBase are abbreviated as $wt$ and $wb$. Workloads with a cache ratio of 10 are labeled as $10X$.

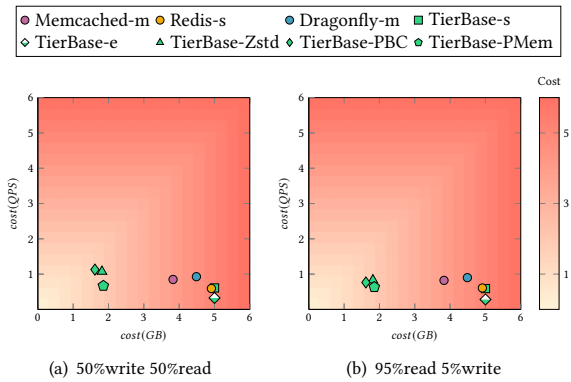
\begin{figure}
\centering

\begin{tikzpicture}[scale=0.6]
\begin{axis}[
    hide axis,
    scale only axis,
    height=0pt,
    width=0pt,
    xmin=0, xmax=1, ymin=0, ymax=1, 
    legend style={font=\fontsize{12}{10}\selectfont, legend columns=4, at={(0.5,0.9)}, anchor=north, row sep=0.5pt, column sep=2pt},
    legend image post style={mark size=2.5pt, only marks}, 
    legend cell align={left},
]

\addlegendimage{mark=*, fill=Purple} \addlegendentry{Memcached-m}
\addlegendimage{mark=*, fill=American Yellow} \addlegendentry{Redis-s}
\addlegendimage{mark=*, fill=Moonstone, draw=black} \addlegendentry{Dragonfly-m}
\addlegendimage{mark=square*, fill=Green2, draw=black} \addlegendentry{TierBase-s}
\addlegendimage{mark=halfsquare*, fill=Green2, draw=black} \addlegendentry{TierBase-e}
\addlegendimage{mark=triangle*, fill=Green2, draw=black} \addlegendentry{TierBase-Zstd}
\addlegendimage{mark=diamond*, fill=Green2, draw=black} \addlegendentry{TierBase-PBC}
\addlegendimage{mark=pentagon*, fill=Green2, draw=black} \addlegendentry{TierBase-PMem}
\end{axis}
\end{tikzpicture}

\par 
\subfigure[50\%write 50\%read]{
\begin{tikzpicture}[scale=0.52]
    \pgfplotsset{
        defaultmark/.style={draw=black, line width=0.25pt},
        mystyle/.style={
            only marks,
            y filter/.code={\pgfmathparse{\pgfmathresult*80000}\pgfmathresult},
            mark size=3pt
        },
        scatter/classes={
            Memcached4={mark=*, fill=Purple},
            Redis1={mark=*, fill=American Yellow},
            TBase1={mark=square*, fill=Green2, draw=black},
            TBaseAdaptive={mark=halfsquare*, fill=Green2, draw=black,mark size=4pt},
            TBaseZstd={mark=triangle*, fill=Green2, draw=black, mark size = 4pt},
            TBasePbc={mark=diamond*, fill=Green2, draw=black,mark size = 4pt},
            TBaseAEP={mark=pentagon*, fill=Green2, draw=black,mark size = 4pt},
            Dragonfly4={mark=*, fill=Moonstone, draw=black}
        }
    }
    \begin{axis}[      
        height=7cm,
        width=7cm,
        xlabel=$cost(GB)$,
        ylabel=$cost(QPS)$,
        xticklabel={
            \pgfmathparse{\tick*10}
            \pgfmathprintnumber{\pgfmathresult}
        },
        yticklabel={
            \pgfmathprintnumber{\tick}
        },
        ytick={1,2,3,4,5,6},
        xtick={0,0.1,0.2,0.3,0.4,0.5,0.6},
        view={0}{90},
        domain=0:0.6,
        y domain=0:6,
        colormap={custom}{
            color(0)=(color1)
            color(1)=(color2)
        },
        samples=20,
        samples y=20,
        scaled y ticks=false,
        scaled x ticks=false,
        yticklabel style={/pgf/number format/fixed},
        xticklabel style={/pgf/number format/fixed},
        legend style={font=\normalsize, at={(1.2,1.3)},legend columns=5},
        legend image post style={mark size=2.5pt},
        legend cell align={left},
    ]
    \addplot3[surf, shader=flat,forget plot] {max(x*10 , y)};
    
    \addplot3[mystyle, defaultmark, scatter, scatter src=explicit symbolic] coordinates {
        (0.382737758,0.0000105987,0) [Memcached4]
        (0.500086233,0.0000076017,0) [TBase1]
        (0.49130154,0.0000074033,0) [Redis1]
        (0.50011367,0.0000040523,0) [TBaseAdaptive]
        (0.181467481,0.0000134159,0) [TBaseZstd]
        (0.161547881,0.0000141014,0) [TBasePbc]
        (0.185388687,0.0000083526,0) [TBaseAEP]
        (0.448874594,0.0000115907,0) [Dragonfly4]
    };
    
    \end{axis}
\end{tikzpicture}
}
\subfigure[95\%read 5\%write]{
\begin{tikzpicture}[scale=0.52]
    \pgfplotsset{
        defaultmark/.style={draw=black, line width=0.25pt},
        mystyle/.style={
            only marks,
            y filter/.code={\pgfmathparse{\pgfmathresult*80000}\pgfmathresult},
            mark size=3pt
        },
        scatter/classes={
            Memcached4={mark=*, fill=Purple},
            Redis1={mark=*, fill=American Yellow},
            TBase1={mark=square*, fill=Green2, draw=black},
            TBaseAdaptive={mark=halfsquare*, fill=Green2, draw=black,mark size=4pt},
            TBaseZstd={mark=triangle*, fill=Green2, draw=black, mark size = 4pt},
            TBasePbc={mark=diamond*, fill=Green2, draw=black,mark size = 4pt},
            TBaseAEP={mark=pentagon*, fill=Green2, draw=black,mark size = 4pt},
            Dragonfly4={mark=*, fill=Moonstone, draw=black}
        }
    }
    \begin{axis}[      
        height=7cm,
        width=7cm,
        xlabel=$cost(GB)$,
        ylabel=$cost(QPS)$,
        xticklabel={
            \pgfmathparse{\tick*10}
            \pgfmathprintnumber{\pgfmathresult}
        },
        yticklabel={
            \pgfmathprintnumber{\tick}
        },
        ytick={1,2,3,4,5,6},
        xtick={0,0.1,0.2,0.3,0.4,0.5,0.6},
        colorbar,
        colorbar style={
            title=Cost,
            ytick={1,3,5},
            yticklabels={$1$,$3$,$5$}
        },
        view={0}{90},
        domain=0:0.6,
        y domain=0:6,
        colormap={custom}{
            color(0)=(color1)
            color(1)=(color2)
        },
        samples=20,
        samples y=20,
        legend style={font=\tiny, legend columns=3},
        legend image post style={mark size=2.5pt},
        legend cell align={left},
        scaled y ticks=false,
        scaled x ticks=false,
        yticklabel style={/pgf/number format/fixed},
        xticklabel style={/pgf/number format/fixed},
    ]
    \addplot3[surf, shader=flat,forget plot] {max(x*10 , y)};
    
    \addplot3[mystyle, defaultmark, scatter, scatter src=explicit symbolic] coordinates {
        (0.382737758,0.0000102787,0) [Memcached4]
        (0.500086233,0.0000073764,0) [TBase1]
        (0.49130154,0.0000075829,0) [Redis1]
        (0.50011367,0.0000035196,0) [TBaseAdaptive]
        (0.181467481,0.0000103843,0) [TBaseZstd]
        (0.161547881,0.0000095431,0) [TBasePbc]
        (0.185388687,0.0000079092,0) [TBaseAEP]
        (0.448874594,0.0000112484,0) [Dragonfly4]
    };
    
    \end{axis}
\end{tikzpicture}
}
\caption{Cost of caching system}
\label{in-memory cost}
\end{figure}
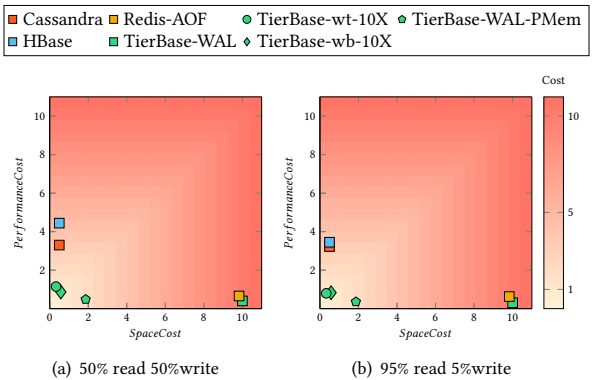
\begin{figure}
\centering

\begin{tikzpicture}[scale = 0.6]
\begin{axis}[
    hide axis,
    scale only axis,
    height=0pt,
    width=0pt,
    xmin=0, xmax=1, ymin=0, ymax=1, 
    legend style={font=\fontsize{12}{10}\selectfont,, legend columns=4, at={(0.5,0.9)}, anchor=north, row sep=0.5pt, column sep=2pt},
    legend image post style={mark size=2.5pt, only marks}, 
    legend cell align={left},
]

\addlegendimage{mark=square*, fill=Giants Orange} \addlegendentry{Cassandra}
\addlegendimage{mark=square*, fill=American Yellow,mark size=4pt} \addlegendentry{Redis-AOF}
\addlegendimage{mark=*, fill=Green2,mark size=4pt} \addlegendentry{TierBase-wt-10X}
\addlegendimage{mark=pentagon*, fill=Green2, draw=black,mark size=4pt} \addlegendentry{TierBase-WAL-PMem}
\addlegendimage{mark=square*, fill=Blue Jeans,mark size=4pt} \addlegendentry{HBase}
\addlegendimage{mark=square*, fill=Green2,mark size=4pt} \addlegendentry{TierBase-WAL}
\addlegendimage{mark=diamond*, fill=Green2,mark size=4pt} \addlegendentry{TierBase-wb-10X}
\end{axis}
\end{tikzpicture}

\par 
\subfigure[50\% read 50\%write]{
\begin{tikzpicture}[scale=0.52]
    \pgfplotsset{
        defaultmark/.style={draw=black, line width=0.25pt},
        mystyle/.style={
            only marks,
            y filter/.code={\pgfmathparse{\pgfmathresult*40000}\pgfmathresult},
            mark size=3.5pt
        },
        scatter/classes={
            Cassandra={mark=square*, fill=Giants Orange},
            hbase={mark=square*, fill=Blue Jeans},
            TBaseAof={mark=square*, fill=Green2},
            RedisAof={mark=square*, fill=American Yellow},
            TBaseWriteThrough10X={mark=*, fill=Green2},
            TBaseWriteBack10X={mark=diamond*, fill=Green2,mark size=5pt},
            TBaseAEP={mark=pentagon*, fill=Green2, draw=black}
        }
    }
    \begin{axis}[      
        height=7cm,
        width=7cm,
        xlabel=$Space Cost$,
        ylabel=$Performance Cost$,
        xticklabel={
            \pgfmathparse{\tick}
            \pgfmathprintnumber{\pgfmathresult}
        },
        yticklabel={
            \pgfmathprintnumber{\tick}
        },
        ytick={2,4,6,8,10},
        xtick={0,2,4,6,8,10},
        view={0}{90},
        domain=0:11,
        y domain=0:11,
        colormap={custom}{
            color(0)=(color1)
            color(1)=(color2)
        },
        samples=20,
        samples y=20,
        scaled y ticks=false,
        scaled x ticks=false,
        yticklabel style={/pgf/number format/fixed},
        xticklabel style={/pgf/number format/fixed},
        legend style={font=\tiny, legend columns=3},
        legend image post style={mark size=2.5pt},
        legend cell align={left},
    ]
    \addplot3[surf, shader=flat,forget plot] {max(x , y)};
    
    \addplot3[mystyle, defaultmark, scatter, scatter src=explicit symbolic] coordinates {
        (0.496651786,0.0000825730,0) [Cassandra]
        (0.496651786,0.0001110618,0) [hbase]
        (10.00172466,0.0000100245,0) [TBaseAof]
        (9.827410515,0.0000164901,0) [RedisAof]
        (0.578500866,0.0000216638748085,0) [TBaseWriteBack10X]
        (0.328457749,0.0000286632,0) [TBaseWriteThrough10X]
        (1.866645295,0.0000120607,0) [TBaseAEP]
    };
    
    \end{axis}
\end{tikzpicture}
}
\subfigure[95\% read 5\%write]{
\begin{tikzpicture}[scale=0.52]
    \pgfplotsset{
        defaultmark/.style={draw=black, line width=0.25pt},
        mystyle/.style={
            only marks,
            y filter/.code={\pgfmathparse{\pgfmathresult*40000}\pgfmathresult},
            mark size=3.5pt
        },
        scatter/classes={
            Cassandra={mark=square*, fill=Giants Orange},
            hbase={mark=square*, fill=Blue Jeans},
            TBaseAof={mark=square*, fill=Green2},
            RedisAof={mark=square*, fill=American Yellow},
            TBaseWriteThrough10X={mark=*, fill=Green2},
            TBaseWriteBack10X={mark=diamond*, fill=Green2,mark size=5pt},
            TBaseAEP={mark=pentagon*, fill=Green2, draw=black}
        }
    }
    \begin{axis}[      
        height=7cm,
        width=7cm,
        xlabel=$Space Cost$,
        ylabel=$Performance Cost$,
        xticklabel={
            \pgfmathparse{\tick}
            \pgfmathprintnumber{\pgfmathresult}
        },
        yticklabel={
            \pgfmathprintnumber{\tick}
        },
        ytick={2,4,6,8,10},
        xtick={0,2,4,6,8,10},
        colorbar,
        colorbar style={
            title=Cost,
            ytick={1,5,10},
            yticklabels={$1$,$5$,$10$}
        },
        view={0}{90},
        domain=0:11,
        y domain=0:11,
        colormap={custom}{
            color(0)=(color1)
            color(1)=(color2)
        },
        samples=20,
        samples y=20,
        legend style={font=\tiny, legend columns=3},
        legend image post style={mark size=2.5pt},
        legend cell align={left},
        scaled y ticks=false,
        scaled x ticks=false,
        yticklabel style={/pgf/number format/fixed},
        xticklabel style={/pgf/number format/fixed},
    ]
    \addplot3[surf, shader=flat,forget plot] {max(x , y)};
    
    \addplot3[mystyle, defaultmark, scatter, scatter src=explicit symbolic] coordinates {
        (0.496651786,0.0000806403,0) [Cassandra]
        (0.496651786,0.0000862645,0) [hbase]
        (10.00172466,0.0000077198,0) [TBaseAof]
        (9.827410515,0.0000157161,0) [RedisAof]
        (0.578500866,0.0000206928530260,0) [TBaseWriteBack10X]
        (0.328457749,0.0000199558,0) [TBaseWriteThrough10X]
        (1.866645295,0.0000088876,0) [TBaseAEP]
        
    };
    
    \end{axis}
\end{tikzpicture}
}
\caption{Cost of database with persistence}
\label{persis cost}
\end{figure}

\subsubsection{Cost analysis for caching systems}
Figure \ref{in-memory cost} presents the cost results for caching systems. The primary cost driver for caching systems is memory storage expenses. Memcached has the lowest storage cost, followed by Dragonfly, while Redis and TierBase without additional features have relatively higher storage costs.

In terms of performance costs, TierBase, Redis, and Memcached exhibit similar low costs in single-thread mode, while Dragonfly shows a higher performance cost. When elastic threading is enabled, TierBase demonstrates improved throughput, leading to a significant reduction in performance costs, nearly half that of single-thread Redis, by efficiently utilizing surplus CPU resources within the containers.

Furthermore, when TierBase employs PMem to extend memory, it achieves a substantial 60\% reduction in storage costs compared to the base configuration, with minimal performance impact. This cost is considerably lower than that of Memcached. Activating compression in TierBase leads to an additional decrease in storage costs.

The results show that TierBase's features can effectively reduce both performance and storage costs. Similar trends are observed across different workload settings, as illustrated in Figure \ref{in-memory cost}(b).

\subsubsection{Cost analysis for databases with persistence}
Figure \ref{persis cost} presents the results for databases with persistence. The traditional key-value store like Cassandra and HBase are observed to have relatively high performance costs while the storage costs are notably low.
For Redis with AOF  and TierBase with WAL, both systems ensure data persistence and adopt a dual-replica strategy for data safety. These approaches result in  lower performance cost but significantly higher storage costs.

TierBase demonstrates a good balance in terms of both performance and storage costs. On the one hand, its inherent characteristics as a caching system enable high throughput. On the other hand, its persistence mechanism does not require storing all data in memory, which contributes to its overall lower costs. It is important to note that under the write-back approach, where data is stored in duplicate copies, the storage cost is higher compared to write-through. However, due to different data update characteristics, the write-back approach exhibits higher throughput in write-intensive scenarios, translating to lower performance costs. This advantage diminishes or even disappears in read-heavy scenarios.

Using PMem for data persistence is a cost-effective choice. Although its space cost is relatively higher compared to write-through and write-back, its performance cost is sufficiently low due to PMem's near-memory speed.

\subsection{Case Study}
\label{case_study}
TierBase is extensively utilized across a wide range of scenarios at Ant Group, with over 3,000 applications leveraging its capabilities. These applications span various use cases, employing hundreds of thousands of CPU cores and several petabytes of memory. Due to space constraints, we will focus on two representative case studies in this paper.

\noindent\textbf{Case 1: User Info Service.}

The User Info Service at Ant Group manages basic user profile data, serving numerous applications through a proprietary SDK with TierBase client. During peak hours on a typical day, this service handles approximately 500,000 updates and 16,000,000 reads per second, indicating a significantly read-heavy workload. Given the service's primary focus on online users, high availability and reliability are of paramount importance.

\subsubsection{Systems comparison}

To assess the cost-effectiveness of various systems, we replayed a real business trace with all databases configured for dual-replica reliability. Figure \ref{case study1} shows that in-memory stores like Redis, Memcached, and Dragonfly have low performance costs but higher storage expenses. TierBase, using compression, halves its original volume, significantly reducing costs compared to Redis. This is particularly advantageous in this read-heavy scenario where performance cost is not primary. The trade-off between performance and storage efficiency demonstrates TierBase's adaptability to specific workload characteristics, optimizing overall cost while maintaining performance efficiency. Activating compression in TierBase yields a 62\% cost reduction compared to TierBase-Raw, showcasing its effectiveness in balancing performance and storage requirements in read-heavy, availability-critical scenarios.

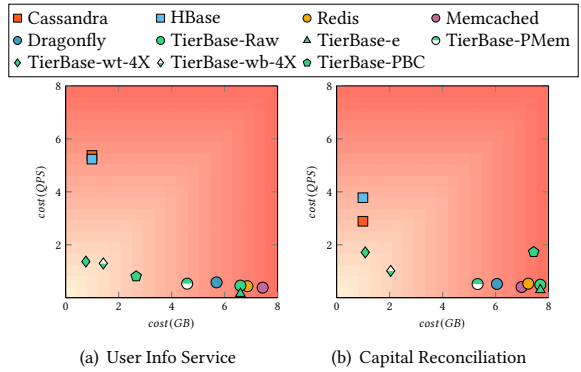
\begin{figure}
\centering

\begin{tikzpicture}[scale = 0.6]
\begin{axis}[
    hide axis,
    scale only axis,
    height=0pt,
    width=0pt,
    xmin=0, xmax=1, ymin=0, ymax=1, 
    legend style={font=\fontsize{12}{10}\selectfont, legend columns=4, at={(0.5,1)}, anchor=north, row sep=0.5pt, column sep=2pt},
    legend image post style={mark size=2.5pt, only marks}, 
    legend cell align={left},
]

\addlegendimage{mark=square*, fill=Giants Orange} \addlegendentry{Cassandra}
\addlegendimage{mark=square*, fill=Blue Jeans,mark size=4pt} \addlegendentry{HBase}
\addlegendimage{mark=*, fill=American Yellow,mark size=4pt} \addlegendentry{Redis}
\addlegendimage{mark=*, fill=Purple, mark size=4pt} \addlegendentry{Memcached}
\addlegendimage{mark=*, fill=Moonstone, draw=black, mark size=4pt} \addlegendentry{Dragonfly}
\addlegendimage{mark=*, fill=Green2, draw=black, mark size=4pt} \addlegendentry{TierBase-Raw}
\addlegendimage{mark=triangle*, fill=Green2,mark size=4pt} \addlegendentry{TierBase-e}
\addlegendimage{mark=halfcircle*, fill=Green2,mark size=4pt} \addlegendentry{TierBase-PMem}
\addlegendimage{mark=diamond*, fill=Green2,mark size=4pt} \addlegendentry{TierBase-wt-4X}
\addlegendimage{mark=halfdiamond*, fill=Green2,mark size=4pt} \addlegendentry{TierBase-wb-4X}
\addlegendimage{mark=pentagon*, fill=Green2, draw=black,mark size=4pt}
\addlegendentry{TierBase-PBC}
\end{axis}
\end{tikzpicture}

\par 
\vspace{-2mm}

\subfigure[User Info Service]{
\label{case study1}
\begin{tikzpicture}[scale = 0.52]
    \pgfplotsset{
        defaultmark/.style={draw=black, line width=0.25pt},
        mystyle/.style={
            only marks,
            y filter/.code={\pgfmathparse{\pgfmathresult*40000}\pgfmathresult},
            mark size=3.5pt,
        },
        scatter/classes={
            Cassandra={mark=square*, fill=Giants Orange},
            Hbase={mark=square*, fill=Blue Jeans},
            Redis={mark=*, fill=American Yellow, mark size=4pt},
            Dragonfly={mark=*, fill=Moonstone, draw=black, mark size=4pt},
            Memcached={mark=*, fill=Purple, mark size=4pt},
            TierBase-raw={mark=*, fill=Green2, draw=black, mark size=4pt},
            TierBase-pmem={mark=halfcircle*, fill=Green2,mark size=4pt},
            TierBase-wb={mark=halfdiamond*, fill=Green2,mark size=4pt},
            TierBase-wt={mark=diamond*, fill=Green2,mark size=4pt},
            TierBase-compression={mark=pentagon*, fill=Green2, draw=black,mark size=4pt},
            TierBase-e={mark=triangle*, fill=Green2,mark size=4pt}
        }
    }
    \begin{axis}[      
        height=7cm,
        width=7cm,
        xlabel=$cost(GB)$,
        ylabel=$cost(QPS)$,
        xticklabel={
            \pgfmathparse{\tick}
            \pgfmathprintnumber{\pgfmathresult}
        },
        yticklabel={
            \pgfmathprintnumber{\tick}
        },
        ytick={2,4,6,8},
        xtick={0,2,4,6,8},
        view={0}{90},
        domain=0:8,
        y domain=0:8,
        colormap={custom}{
            color(0)=(color1)
            color(1)=(color2)
        },
        samples=20,
        samples y=20,
        scaled y ticks=false,
        scaled x ticks=false,
        yticklabel style={/pgf/number format/fixed},
        xticklabel style={/pgf/number format/fixed},
        legend style={font=\tiny, legend columns=3,
        row sep=0.5pt,
        column sep=2pt
        },
        legend cell align={left},
        legend image post style={mark size=2.5pt},
    ]
    \addplot3[surf, shader=flat,forget plot] {max(x, y)};
    
    \addplot3[mystyle, defaultmark, scatter, scatter src=explicit symbolic] coordinates {
        (0.993303571,0.000134282,0) [Cassandra]
        (0.993303571,0.000130839,0) [Hbase]
        (2*3.430351257,0.000010850,0) [Redis]
        (2*2.847989082,0.000014454,0) [Dragonfly]
        (2*3.718143463,0.000009500,0) [Memcached]
        (2*3.297821045,0.000011431,0) [TierBase-raw]
        (0.762985343,0.000034125,0) [TierBase-wt]
        (1.422549552,0.000032339,0) [TierBase-wb]
        (2*1.331118,0.00002016,0)[TierBase-compression]
        (2*2.293272972,0.000013222,0) [TierBase-pmem]
        (2*3.297821045,0.000003850,0) [TierBase-e]

    };
    
    \end{axis}
\end{tikzpicture}
}
\subfigure[Capital Reconciliation]{
\label{case study2}
\begin{tikzpicture}[scale = 0.52]
    \pgfplotsset{
        defaultmark/.style={draw=black, line width=0.25pt},
        mystyle/.style={
            only marks,
            y filter/.code={\pgfmathparse{\pgfmathresult*40000}\pgfmathresult},
            mark size=3.5pt,
        },
        scatter/classes={
            Cassandra={mark=square*, fill=Giants Orange},
            Hbase={mark=square*, fill=Blue Jeans},
            Redis={mark=*, fill=American Yellow, mark size=4pt},
            Dragonfly={mark=*, fill=Moonstone, draw=black ,mark size=4pt},
            Memcached={mark=*, fill=Purple, mark size=4pt},
            TierBase-raw={mark=*, fill=Green2, draw=black, mark size=4pt},
            TierBase-pmem={mark=halfcircle*, fill=Green2,mark size=4pt},
            TierBase-wb={mark=halfdiamond*, fill=Green2,mark size=4pt},
            TierBase-compression={mark=pentagon*, fill=Green2, draw=black,mark size=4pt},
            TierBase-wt={mark=diamond*, fill=Green2,mark size=4pt},
            TierBase-e={mark=triangle*, fill=Green2,mark size=4pt}
        }
    }
    \begin{axis}[      
        height=7cm,
        width=7cm,
        xlabel=$cost(GB)$,
        ylabel=$cost(QPS)$,
        xticklabel={
            \pgfmathparse{\tick}
            \pgfmathprintnumber{\pgfmathresult}
        },
        yticklabel={
            \pgfmathprintnumber{\tick}
        },
        ytick={2,4,6,8},
        xtick={0,2,4,6,8},
        view={0}{90},
        domain=0:8,
        y domain=0:8,
        colormap={custom}{
            color(0)=(color1)
            color(1)=(color2)
        },
        samples=20,
        samples y=20,
        scaled y ticks=false,
        scaled x ticks=false,
        yticklabel style={/pgf/number format/fixed},
        xticklabel style={/pgf/number format/fixed},
        legend style={font=\tiny, legend columns=3,
        row sep=0.5pt,
        column sep=2pt
        },
        legend cell align={left},
        legend image post style={mark size=2.5pt},
    ]
    \addplot3[surf, shader=flat,forget plot] {max(x, y)};
    
    \addplot3[mystyle, defaultmark, scatter, scatter src=explicit symbolic] coordinates {
        (2*3.491230011,0.000010061,0) [Memcached]
        (2*3.61715126,0.000013230,0) [Redis]
        (2*3.025972366,0.000012971,0) [Dragonfly]
        (0.993303571,0.000072177,0) [Cassandra]
        (0.993303571,0.000094415,0) [Hbase]
        (2*3.844118118,0.000012241,0) [TierBase-raw]
        (1.081582783,0.000042850,0) [TierBase-wt]
        (2.042612312,0.000025299,0) [TierBase-wb]
        (2*3.719650269,0.000043146,0) [TierBase-compression]
        (2*2.661312543,0.000012898,0) [TierBase-pmem]
        (2*3.844118118,0.000007357,0) [TierBase-e]
    };
    
    \end{axis}
\end{tikzpicture}
}
\caption{Cost of case study}
\end{figure}

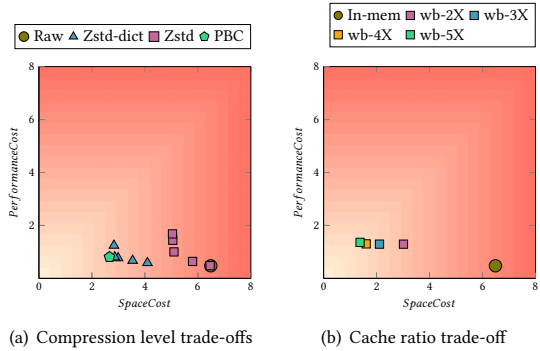
\begin{figure}
\centering





\subfigure[Compression level trade-offs]{
\label{compression_tradeoff}
\begin{tikzpicture}[scale = 0.52]
    \pgfplotsset{
        defaultmark/.style={draw=black, line width=0.25pt},
        mystyle/.style={
            only marks,
            y filter/.code={\pgfmathparse{\pgfmathresult*40000}\pgfmathresult},
            mark size=3.5pt,
        },
        scatter/classes={
            Full={mark=*, fill=Olive, mark size=4.5pt},
            L22={mark=triangle*, fill=Moonstone, draw=black, mark size=4pt},
            L1nd={mark=square*, fill=Purple, mark size=3pt},
            PBC={mark=pentagon*, fill=Green2, draw=black,mark size=4pt}
        }
    }
    \begin{axis}[      
        height=7cm,
        width=7cm,
        xlabel=$Space Cost$,
        ylabel=$Performance Cost$,
        xticklabel={
            \pgfmathparse{\tick}
            \pgfmathprintnumber{\pgfmathresult}
        },
        yticklabel={
            \pgfmathprintnumber{\tick}
        },
        ytick={2,4,6,8},
        xtick={0,2,4,6,8},
        view={0}{90},
        domain=0:8,
        y domain=0:8,
        colormap={custom}{
            color(0)=(color1)
            color(1)=(color2)
        },
        samples=20,
        samples y=20,
        scaled y ticks=false,
        scaled x ticks=false,
        yticklabel style={/pgf/number format/fixed},
        xticklabel style={/pgf/number format/fixed},
        legend style={font=\fontsize{12}{10}\selectfont,
        legend columns=4,
        at={(1,1.2)},
        row sep=0.5pt,
        column sep=2pt
        },
        legend cell align={left},
        legend image post style={mark size=3.5pt},
    ]
    \addplot3[surf, shader=flat,forget plot] {max(x, y)};
    
    \addplot3[mystyle, defaultmark, scatter, scatter src=explicit symbolic] coordinates {
        (2*3.243796 ,0.00001189,0) [Full]
        (2*2.048442,0.00001465,0) [L22] 
        (2*1.768524,0.00001683,0) [L22] 
        (2*1.495251,0.00001918,0) [L22] 
        (2*1.431094,0.00002005,0) [L22] 
        (2*1.415869,0.00003107,0) [L22]
        (2*3.231505,0.000012,0) [L1nd] 
        (2*2.9,0.000016,0) [L1nd]
        (2*2.548791,0.000025,0) [L1nd] 
        (2*2.526948,0.000036,0) [L1nd] 
        (2*2.521176,0.000042,0) [L1nd] 
        (2*1.331118,0.00002016,0) [PBC]
    };
    
    \legend{Raw,Zstd-dict,Zstd,PBC}
    \end{axis}
\end{tikzpicture}
}
\hspace{1mm}
\subfigure[Cache ratio trade-off ]{
\label{Tier_tradeoff}
\begin{tikzpicture}[scale = 0.52]
    \pgfplotsset{
        defaultmark/.style={draw=black, line width=0.25pt},
        mystyle/.style={
            only marks,
            y filter/.code={\pgfmathparse{\pgfmathresult*40000}\pgfmathresult},
            mark size=3.5pt,
        },
        scatter/classes={
            Full={mark=*, fill=Olive, mark size=4.5pt},
            WB2X={mark=square*, fill=Purple, mark size=3pt},
            WB3X={mark=square*, fill=Moonstone, draw=black, mark size=3pt},
            WB4X={mark=square*, fill=American Yellow, mark size=3pt},
            WB5X={mark=square*, fill=Green2, draw=black, mark size=3pt}
        }
    }
    \begin{axis}[      
        height=7cm,
        width=7cm,
        xlabel=$Space Cost$,
        ylabel=$Performance Cost$,
        xticklabel={
            \pgfmathparse{\tick}
            \pgfmathprintnumber{\pgfmathresult}
        },
        yticklabel={
            \pgfmathprintnumber{\tick}
        },
        ytick={2,4,6,8},
        xtick={0,2,4,6,8},
        view={0}{90},
        domain=0:8,
        y domain=0:8,
        colormap={custom}{
            color(0)=(color1)
            color(1)=(color2)
        },
        samples=20,
        samples y=20,
        scaled y ticks=false,
        scaled x ticks=false,
        yticklabel style={/pgf/number format/fixed},
        xticklabel style={/pgf/number format/fixed},
        legend style={
        font=\fontsize{12}{10}\selectfont,
        legend columns=3,
        at={(1,1.3)},
        legend columns=3,
        row sep=0.5pt,
        column sep=2pt
        },
        legend cell align={left},
        legend image post style={mark size=3pt},
    ]
    \addplot3[surf, shader=flat,forget plot] {max(x, y)};
    
    \addplot3[mystyle, defaultmark, scatter, scatter src=explicit symbolic] coordinates {
        (1.607845231,0.00003258,0) [WB4X]
        (2.110432121,0.00003234,0) [WB3X]
        (3.015521338,0.00003231,0) [WB2X]
        (6.487591315,0.00001189,0) [Full]
        (1.376655262,0.0000340,0) [WB5X]
    };
    
    \legend{In-mem,wb-2X,wb-3X,wb-4X,wb-5X}
    \end{axis}
\end{tikzpicture}
}
\caption{Space-Performance Cost Trade-offs}
\label{Trade-off}
\end{figure}

\subsubsection{Space-performance cost trade-offs}

We demonstrate the trade-off using our proposed cost model under use case 1. Figure \ref{Trade-off} shows the workload's space cost significantly exceeds its performance cost. We employ compression techniques, sacrificing some performance to save substantial space. We tested Zstd compression levels -50, -10, 1, 15, and 22, both with and without a dictionary. As shown in Figure \ref{Trade-off}(a), higher compression levels increase space savings but have an upper bound, beyond which compression ratio gains become marginal while performance costs grow considerably. Practically, we may select compression level 1 for better performance tolerance. Additionally, pre-trained compression yields more substantial cost savings compared to compression without pre-training.

We also evaluate the cost-effectiveness of TierBase with write-back policy using four cache ratios ranging from 2X to 5X. As shown in Figure \ref{Trade-off}(b), Higher cache ratios result in lower space costs but higher performance costs. 

The result reveals that a cache ratio of 5X approximately achieves the optimal balance between performance and storage costs as predicted by our model. The results validate our cost model for accurately guiding cost optimization, confirming that the achieved effects align with expectations.




\subsubsection{Break-even interval}

Furthermore, we calculate several sets of break-even intervals between the fast and slow storage configurations on TierBase based on the analyses in Section \ref{interval}. As shown in Table \ref{interval_table}, if the average access interval for a key in the workload is less than 98 seconds, the default TierBase is the most cost-effective choice. For access intervals between 98 and 264 seconds, TierBase with PMem mode is recommended. When the average access interval exceeds 264 seconds, employing compression becomes the optimal solution. By collecting the average access interval for a key in the real workload, we observe that it exceeds 1018 seconds. Consequently, TierBase is employed as a single-layer caching system, leveraging pre-trained compression (PBC) to optimize memory usage. This approach achieves a 25\% compression rate for values and realizes cost savings of 50\%, which is significant considering the hundreds of thousands of CPU cores utilized in this case.


\begin{table}[htbp]
\centering
\caption{Break-even interval between different configurations.}
\begin{tabular}{ccc}
\toprule
Fast Storage & Slow Storage & Time Interval(s) \\
\midrule
Raw          & PMem         & 98 \\
Raw          & Compression(PBC) & 184 \\
PMem         & Compression(PBC) & 264 \\
\bottomrule
\end{tabular}
\label{interval_table}
\end{table}

Although write-through caching could potentially be applicable in this scenario and achieve a 60\% cost reduction, the decision to prioritize compression over write-through caching is driven by the client's stringent requirements for low latency and high stability when serving online requests. TierBase's adaptability allows it to configure for aligning with the specific needs and priorities of each client, ensuring an optimal balance between cost efficiency and performance in real-world scenarios.




\noindent\textbf{Case 2: Capital Reconciliation.}
TierBase is also deployed in the capital reconciliation business at Ant Group. As a risk control scenario focused on financial auditing and verification, the capital reconciliation business is particularly sensitive to costs. During peak shopping seasons, the overall QPS for capital reconciliation can reach tens of millions. TierBase's write-through and write-back caching strategies are selectively employed based on different scenarios. For this case study, we choose one of the main scenarios where the read and write operations are close to a 1:1 ratio. In this scenario, data from different channels is written into TierBase and then read out by the reconciliation system for verification.


Figure \ref{case study2} illustrates the cost breakdown for this scenario. Disk-based key-value stores like HBase and Cassandra exhibit low space and performance costs. When TierBase is configured with write-through, performance costs are lowered by 35\% compared to Cassandra. In high-throughput scenarios, enabling write-back mode further enhances performance. With the same configuration, TierBase can achieve 2.6x the performance of Cassandra. Overall, TierBase reduces costs by at least 37\% compared to both Cassandra and HBase. Additionally, it cuts costs by 70\% compared to TierBase default configuration. 

The observations in the capital reconciliation case study reveal that recent data is frequently accessed in the cache, while long-term data is occasionally retrieved. Online statistics shows that TierBase with write-through mode achieves a cache hit rate of approximately 80\%, with only 1\% of the hottest data stored in the cache tier. This demonstrates the effectiveness of TierBase's cache-storage disaggregation in significantly reducing costs for workloads with temporal access skewness.


\
\newline
\indent

\vspace{-7mm} 
\section{Related Work}  

\subsection{Key-Value Stores}


Key-value stores, a type of NoSQL databases, play a crucial role in Internet applications. Diverging from traditional relational databases, they provide fast, scalable, and efficient data access essential for a wide range of online applications.

Memcached \cite{memcached} is a distributed caching system enhancing web application performance by minimizing database load and optimizing memory across servers. 
Redis \cite{redis}, unlike Memcached's focus on caching, is a multifunctional in-memory database with a wide variety of data types, extended functionality, and persistence, suitable for complex, high-performance applications.
KeyDB \cite{keydb} enhances Redis with multi-threading. Dragonfly \cite{dragonflydb} offers compatibility with Redis and Memcached, using multi-threaded, shared-nothing architecture. Etcd \cite{etcd} is crucial for Kubernetes, providing consistent configuration across clusters. EVCached \cite{evcache}, ElastiCache \cite{elasticache}, and Azure Cache for Redis \cite{azurecache} are managed caching solutions, improving cloud applications by reducing database loads and enabling quick data access.

\subsection{Cost Optimization}

\subsubsection{Cost model} Constructing a reasonable cost model is crucial for the optimization of database costs. Total Cost of Ownership (TCO)\cite{tco} refers to the sum of all related costs throughout the entire lifecycle of an item or service, including purchasing, operating, maintaining, and disposing of it. This concept is widely applied across various fields. Some cost models \cite{10.5555/1287369.1287387, 10.1007/s00778-021-00660-x, 9835426}  primarily focus on estimating or optimizing query execution time or overall system performance.  They consider factors like memory access patterns, cache behavior, and data structure efficiency to predict or improve query processing speed. Other cost models \cite{10.14778/3551793.3551837, 10.1145/2213836.2213844} explicitly consider financial costs, such as the cost of operating cloud resources or the cost differences between different types of storage medium. 
However, these methods often struggle to adequately represent the complex relationships between configurations and costs, especially when considering diverse workload characteristics. As a result, they are not directly applicable to our study's objective.

\subsubsection{Optimization strategy} Contrasting with TierBase, which concentrates on a comprehensive cost model that balances performance and storage expenses, some studies \cite{anna, 10.1145/1989323.1989357, 8741093, 10597997}  achieve more effective resource allocation or configuration parameter adjustments by monitoring workload variations within databases, thereby enhancing resource utilization and reducing costs. Cosine \cite{cosine} actively tailors its storage engine architecture to optimize costs, adapting to workload demands, cloud budgets, and specific performance objectives.

\subsubsection{Compression}
In database technology, Lempel-Ziv(LZ) algorithms are widely used for data compression to improve storage and transfer efficiency. TiDB \cite{tidb} and Hadoop \cite{hadoop} use LZ4, while LevelDB \cite{leveldb} uses Snappy \cite{snappy}. Zstandard (Zstd) \cite{zstd} is chosen by Facebook's RocksDB \cite{rocksdb} and Redshift \cite{redshift} for its efficiency. Apache Cassandra \cite{Cassandra} also uses compression to reduce disk space use. For memory efficiency, SlimCache \cite{SlimCache} and zExpander \cite{zExpander} compress data in caches, and COLLATE \cite{COLLATE} applies lightweight compression in in-memory databases.

\subsubsection{Thread management}
Thread management techniques play a pivotal role in DBMSs by optimizing the handling of concurrent client requests and maximizing resource utilization.
In systems like \cite{rocksdb,riak,oracle,mysql}, dynamic adjustment of thread pools is utilized to efficiently handle fluctuating workloads. On the other hand, databases such as \cite{memcached,redis} maintain a fixed thread count, ensuring consistent resource allocation. \cite{tarantool,mongodb, coroutines1,coroutines2}, use coroutines in certain scenarios for managing concurrent requests, balancing thread overhead with concurrency. \cite{numa_aware_thread} discusses thread management in the NUMA-aware task scheduling. \cite{aerospike} merges these approaches, applying dynamic thread pools for general requests and fixed threads for specific tasks.

\subsubsection{Persistent memory}
Non-volatile memory, also known as persistent memory is a type of memory that keeps its data even when the power is turned off. 
Many current efforts are integrating PMem into storage systems \cite{Viper,Tair-PMem,ChameleonDB,ByteDance_pmem,NV-SQL,lsmnvm,zen} to enhance performance. However, PMem has its own set of challenges and issues. Some studies \cite{costmodel_pmem,PerMA-Bench,manage_nvmsystem} are exploring how systems should navigate between PMem and DRAM, addressing concerns about optimal choice and management.


\section{Conclusion}



This paper introduces TierBase, a distributed key-value store developed by Ant Group to address the challenges of managing large-scale online data serving systems. TierBase employs a tiered storage architecture and incorporates cost-effective optimizations. We propose a space-performance cost model to guide the selection of optimal storage configurations for diverse workloads. Extensive evaluations using both synthetic benchmarks and real-world workloads demonstrate TierBase's superior cost-effectiveness compared to existing solutions.

\section*{Acknowledgment}

We thank everyone who contributed to the design and development of the TierBase system and its cost model. This work was supported by Ant Group Research Fund.


\bibliographystyle{plain} 
\bibliography{references}

\end{document}